\documentclass[lettersize,journal]{IEEEtran}
\usepackage{amsmath,amsfonts}
\usepackage{algorithmic}
\usepackage{algorithm}
\usepackage{array}
\usepackage[caption=false,font=normalsize,labelfont=sf,textfont=sf]{subfig}
\usepackage{textcomp}
\usepackage{stfloats}
\usepackage{url}
\usepackage{verbatim}
\usepackage{graphicx}
\usepackage{cite}
\usepackage{stfloats}
\usepackage{threeparttable}
\usepackage{makecell}
\usepackage{CJK}
\usepackage{enumitem}
\usepackage{lipsum}
\allowdisplaybreaks[4]

\makeatletter

\newcommand{\Rmnum}[1]{\expandafter\@slowromancap\romannumeral #1@}
\makeatother

\usepackage{epstopdf}
\usepackage{caption}

\usepackage{amsthm}

\usepackage{booktabs}
\DeclareCaptionLabelFormat{cont}{#1~#2\alph{ContinuedFloat}}
\newtheorem{definition}{Definition}
\newtheorem{theorem}{Theorem}

\newtheorem{remark}{Remark}
\newtheorem{assumption}{Assumption}

\begin{document}
\title{Distributed Affine Formation Control of Linear Multi-agent Systems with Adaptive Event-triggering}
\author{Chenjun Liu, Jason J. R. Liu,~\IEEEmembership{Member,~IEEE,} Zhan Shu,~\IEEEmembership{Senior~Member,~IEEE,} and James Lam,~\IEEEmembership{Fellow,~IEEE}
\thanks{This work is supported in part by the National Natural Science Foundation of China under Grant 62403008, the Macao Science and Technology Development Fund under Grant 0008/2023/ITP1, Grant 0139/2023/RIA2, and the University of Macau under Grant SRG2022-00054-FST and Grant MYRG-CRG2024-00037-FST-ICI.}
\thanks{Chenjun Liu and Jason J. R. Liu are with the Department of Electromechanical Engineering, University of Macau, Macau (email: yc37973@um.edu.mo (C. Liu);  jasonliu@um.edu.mo (J. Liu)). }
\thanks{Zhan Shu is with the Department of Electrical and Computer Engineering, University of Alberta, Edmonton, Alberta, T6G 1H9, Canada (email: zhan.shu@ualberta.ca (Z. Shu)).}
\thanks{James Lam is with the Department of Mechanical Engineering, University of Hong Kong, Pokfulam Rd, Hong Kong (email: james.lam@hku.hk (J. Lam)).}
\thanks{Corresponding author: Jason J. R. Liu.}
}
\markboth{Journal of \LaTeX\ Class Files,~Vol.~14, No.~8, August~2021}%
{Shell \MakeLowercase{\textit{et al.}}: A Sample Article Using IEEEtran.cls for IEEE Journals}


\maketitle

\begin{abstract}
Concerning general multi-agent systems with limited communication, this paper proposes distributed formation control protocols under adaptive event-triggered schemes to operate affine transformations of nominal formations. To accommodate more practical system mechanics, we develop an event-triggered controller that drives the leader to a desired state by bringing in the compensation term. Based on triggering instants' state information, an affine formation control method with adaptive event-triggering is designed for each follower, making the whole protocol effective in refraining from successive communication while not relying on predefined global information. In particular, mitigating the effect of partial state availability, an output-based control solution is presented to expand the protocol's serviceable range. Finally, we perform numerical simulations on the formation and its affine transformations to verify the effectiveness of the control protocol and the feasibility of the event-triggered mechanism.
\end{abstract}

\begin{IEEEkeywords}
Adaptive event-triggered control, affine formation control, output-based control, Zeno behavior
\end{IEEEkeywords}

\section{Introduction}
\IEEEPARstart{M}{ulti-agent} systems (MASs) consist of individual entities, by exploiting intercommunication and self-organization capabilities, thus rendering an orderly and synergistic movement. Formation control of MASs firstly intends to drive agents from scattered and unorganized states to form a prescribed paradigm, having been applied in marine search and rescue \cite{Zhang2022Event}, agricultural production \cite{Chen2022Information}, electric power inspection \cite{Luo2023A} and others. Moreover, to flexibly cope with unknown circumstances, the deformation reconfiguration technique of formation has attracted numerous researchers.

Conventional formation control methods have been categorized by distinct sensing capabilities and topological imperatives. When relative displacement is measured, MASs can realize translation formation by tracking the time-varying signals (state/velocity) \cite{Marina2020Maneuvering}. Also, if the rotated or otherwise transformed formations belong to the feasible set of time-varying formations, they may be achieved \cite{Zhang2024Adaptive}, subject to the fact that the desired state of each agent is known and needs to be transmitted. Whereas in distance-based formation control, translational and rotational transformations are implemented \cite{Mehdifar2020Prescribed}. In contrast, Reference 
\cite{Su2023Bearing} designed time-varying bearings covering translation, scaling, and rotation transformations, with only known leaders' desired states. It is noted that the above methods cannot fulfill more general transformations including shear, coplanar, and others.

To simultaneously realize various transformations, several novel methods have been proposed, including hybrid measurements \cite{Fang2024Integrated}, barycentric coordinates \cite{Han2018A}, complex Laplacian matrices \cite{Fang2024Distributed}, and stress matrices. In particular, the stress matrix is invariant to arbitrary affine transformations, with the potential to realize more general formation transformations. References \cite{Lin2015Necessary,Zhao2018Affine} first investigated static affine formation and formation maneuver control theories focusing on low-order integrator dynamical models. Then, several works studying the finite-time performance and proposing solutions to the system uncertainty have been published in \cite{Zhao2024Specified} and \cite{Zhi2021Leader}. Meanwhile, some system models were characterized by high-order integrator dynamics  \cite{Chen2020Distributed,Xu2020Affine} and more general linear dynamics \cite{Onuoha2020Fully,Chang2023Fully,Xu2018Affine,Xu2019Affine}. However, References \cite{Onuoha2020Fully,Chang2023Fully,Xu2018Affine} all assume that the leaders' states remain desired, and Reference \cite{Xu2019Affine} assumes that the system matrix is Hurwitz to ensure that the leader moves to the desired state, which both limit the mechanical structure of the actual system. Hence, this paper investigates the leader's control input settings to solve it.

It is worth mentioning that most of the above results rely on continuous communication, restricting the usability in bandwidth and energy-constrained settings. Event-triggered control (ETC) provides an effective solution by performing only the necessary information transmissions and control updates. While on the stress matrix-based formation, Reference \cite{Yang2020Affine} studied the event-triggered affine formation maneuver control problem for MASs described by single-integrator dynamics, but the fixed triggering threshold led to a bounded formation error, reducing the system performance. The sampled-data-based control scheme was proposed in \cite{Ma2023Event} to settle the affine formation issue for second-order MASs. Still, global information was required to determine the triggering function. Then in \cite{Chang2023Fully}, to modulate the interaction instant during affine formations of linear MASs, the dynamic triggering threshold depended only on the measurement error itself was adopted. Synthesizing the above concerns, developing more efficient event-triggered mechanisms is necessary, such as avoiding global information with adaptive control and moderating the triggering performance with the formation errors.

Inspired by the above studies, this paper investigates the adaptive event-triggered affine formation control (AFC) issue for general linear MASs. The major contributions are elaborated as follows:
\begin{enumerate}
\item [1)] This paper proposes effective event-triggered AFC protocols to regulate MASs based only on the information at discontinuous event-triggering instants, preventing network congestion. Different from existing results \cite{Lin2015Necessary,Zhao2018Affine,Zhao2024Specified,Zhi2021Leader}, we consider the MASs featuring more general linear dynamics. Compared with linear MASs studies   \cite{Onuoha2020Fully,Chang2023Fully,Xu2018Affine,Xu2019Affine}, leaders are driven to move to the desired positions by an event-triggered controller without assuming the mechanical structure of the actual system.
\item [2)] Adaptive techniques are utilized to affine formation controllers and event-triggered functions to modulate gains dynamically. Unlike \cite{Ma2023Event}, it has the advantage of not relying on a priori global topological information such as the eigenvalues of the Laplacian matrix. Even if there are variables in the MAS components, notably the number of agents or the communication links, the protocol can also be applied in an extended way.
\item [3)] We consider various factors in reality, such as the absence of sensors or damaged sensors, which may result in the full state of the system not being directly available. This paper extends the scope of application by giving the output-based protocol for affine formation.
\end{enumerate}

The remainder of this paper is organized as follows. Section \ref{Preliminaries} provides some essential concepts and the pending issue. In Section \ref{MainResults}, we present state/output-based event-triggered formation methods. Section \ref{Simulations} verifies the above results by simulations. In Section \ref{Conclusions}, the paper summarizes the outcomes and highlights prospective research directions. 

\noindent {\bf Notations:} The symbol $\|\cdot\|$ indicates either the spectral norm of a matrix or the Euclidean norm of a vector. For matrices $A_{i} \in \mathbb{R}^{m \times m}$, ${\rm diag}\{A_{1}, \dots, A_{n}\}$ denotes the partitioned diagonal matrix and $A_{i} \otimes A_{j}$ is the Kronecker product.  $\lambda(A_{i})$ and $\lambda_{\min}(A_{i})$ mean the eigenvalues and smallest of them. $\mathbf{1}_{n} = [1, \dots, 1]^{T}\in \mathbb{R}^{n}$, $\mathbf{0}_{n}= [0, \dots, 0]^{T}\in \mathbb{R}^{n}$ and ${I}_{n} = \mathrm{diag}\{1, \dots, 1\}\in \mathbb{R}^{n\times n}$ denote the all-one column vector, the all-zero column vector and the unit matrix, respectively. $\iota=\sqrt{-1}$ represents the imaginary unit. For constants $b_{i} \in \mathbb{R}$, where $i = 1, \dots, n$, $\max_{i\in\{1, \dots, n\}}\{b_{i}\}$ or $\min_{i\in\{1, \dots, n\}}\{b_{i}\}$ is the maximum/minimum value.

\section{Preliminaries \label{Preliminaries}}
\subsection{Graph Theory}
An undirected graph $G=(V,E)$ is to describe the communication topology between agents, having a node set $V=\{1,\dots,n\}$ and an edge set $E\subseteq V\times V$. For the leader-follower mode, $V$ is divided by the leader set $V_{l}=\{1,\dots,n_{l}\}$ and the follower set $V_{f}=\{n_{l}+1,\dots,n\}$. The $i$-th node can perceive the information of $j$-th node if $(i,j)\in E$, so the $j$-th node's neighbors' set is designated as $\mathcal{N}_{j}=\{o\in V : (j,o)\in E\}$.

The definition of stress matrix is imported from Reference \cite{Lin2015Necessary}. For an MAS having $n$ agents, the configuration is $p=[p_{l}^{T},p_{f}^{T}]^{T}=[p^{T}_{1},\dots,p^{T}_{n}]^{T}$, where $p_{i}\in \mathbb{R}^{d}$ with $d>2$ denotes the position. The formation $(G, p)$ is defined by mapping the $i$-th node of $G$ to $p_{i}$. To assign a weight $w_{ij}\in \mathbb{R}$ to edge $(i,j)\in E$, the set of scalars $\{w_{ij}\}_{(i,j)\in E}$ is labeled as the stress of $(G, p)$. The equilibrium stress can convert to $(\Omega\otimes I_{d}) p=0$, where the stress matrix $\Omega\in \mathbb{R}^{n\times n}$ is denoted as $[\Omega]_{ii}=\sum_{j\in \mathcal{N}_{i}}w_{ij}$, $[\Omega]_{ij}=-w_{ij}$ for $j\in \mathcal{N}_{i}$ and 0 otherwise for $j\notin \mathcal{N}_{i}$. For the leader-follower mode, $\Omega$ can be partitioned as $\Omega=[\Omega_{ll},\Omega_{lf};\Omega_{fl},\Omega_{ff}]$,
where $\Omega_{ll}\in\mathbb{R}^{n_{l}\times n_{l}}$ and $\Omega_{ff}\in\mathbb{R}^{n_{f}\times n_{f}}$.

\subsection{Affine Formation}
Network localizability is the foundation for achievable formation tasks. To explicate affine localizability, the affine span $\mathcal{S}$ of a series of nodes $\{a_{i}\}^{n}_{i=1}$ in $\mathbb{R}^{d}$ is first stated as
\begin{equation*}
\mathcal{S}=\left\{\sum_{i=1}^{n}{b_{i}a_{i}|\ b_{i}\in \mathbb{R} \ \mathrm{for}\ \forall i}\ \mathrm{and} \ \sum_{i=1}^{n}{b_{i}}=1 \right\}.
\end{equation*}
If there are non fully zeroed $\{b_{i}\}^{n}_{i=1}$ satisfying $\sum_{i=1}^{n}{b_{i}a_{i}}=0$ and $\sum_{i=1}^{n}{b_{i}}=0$, these nodes are affinely dependent, conversely affinely independent. The configuration $P(a)\in \mathbb{R}^{n\times d}$ and the  augmented matrix $\tilde{P}(a)\in \mathbb{R}^{n\times(d+1)}$ are introduced by $P(a)=[a^{T}_{1};\dots;a^{T}_{n}]$ and $\bar{P}(a)=[a^{T}_{1},1;\dots;a^{T}_{n},1]$, where $\{a_{i}\}^{n}_{i=1}$ are affinely independent if and only if rows of $\bar{P}(a)$ is linearly independent, meaning that $\bar{P}^{T}(a)b = 0$ has a unique zero solution $b=\mathbf{0}_{n}$. Since $\bar{P}(a)$ has $d+1$ columns, then at most $d+1$ nodes are affinely independent in $\mathbb{R}^{d}$.

Then, we obtain the definition of the target affine formation, which is a time-varying or constant affine transformation of the nominal configuration $(G,a)$.
\begin{definition}\label{d1}
\cite{Zhao2018Affine} The target affine formation satisfies
\begin{equation*}
p^{*}(t)=[I_{n}\otimes \Upsilon(t)]a+\mathbf{1}_{n}\otimes \upsilon(t),
\end{equation*}
where $\Upsilon(t)$ and $\upsilon(t)$ are continuous of $t$, and $p^{*}(t)=[(p_{l}^{*}(t))^{T},(p_{f}^{*}(t))^{T}]^{T}$.
\end{definition}

\begin{assumption}\label{a1}
Assume that a set of leader nodes $\{a_{i}\}^{n_{l}}_{i=1}$ of the nominal configuration $(G,a)$ affinely spans $\mathbb{R}^{d}$, which means there are at least $d+1$ leaders.
\end{assumption}

\begin{assumption}\label{a2}
The stress matrix $\Omega$ of the nominal configuration $(G,a)$ is positive semi-definite with $\mathrm{rank}(\Omega)=n-d-1$.
\end{assumption}

\begin{remark}
When Assumptions \ref{a1} and \ref{a2} hold, the sufficient and necessary condition for affine localization is that $\Omega_{ff}$ is nonsingular, which implies $\Omega_{ff}$ is positive definite and the states of followers can be obtained by $p_{f}=-(\Omega^{-1}_{ff}\Omega_{fl}\otimes I_{d})p_{l}$.
\end{remark}

\subsection{Problem Statement}
Consider the general linear MAS consisting of $n_{l}$ leaders and $n_{f}$ followers ($n=n_{l}+n_{f}$), characterized as follows:
\begin{equation} \label{equation:1}
\begin{cases}
\dot{p}_{i}(t)=A{p}_{i}(t)+Bu_{i}(t),   \\
q_{i}(t)=Cp_{i}(t), \quad i\in V,
\end{cases}
\end{equation}
where $p_{i}(t)\in \mathbb{R}^{d}$, $u_{i}(t)\in \mathbb{R}^{m}$ and $q_{i}(t)\in \mathbb{R}^{r}$ denote the state, the control input and the output of the $i$-th agent, respectively. $A\in \mathbb{R}^{d\times d}$, $B\in \mathbb{R}^{d\times m}$ and $C \in \mathbb{R}^{r\times d}$ are constant matrices with ${\rm rank}(B) = m$.

With limited communication, the control objective of this paper is to design controllers such that leaders move to the desired position, i.e., $\mathrm{lim}_{t\rightarrow\infty}(p_{l}(t)-p^{*}_{l}(t))=0$, while followers manage to follow to realize the affine formation, that is, $\mathrm{lim}_{t\rightarrow\infty}[p_{f}(t)+(\Omega^{-1}_{ff}\Omega_{fl}\otimes I_{d})p_l(t)]=0$. In the sense, the control input is designed using merely the state information of neighbors at the triggering instant sequence $\{t_{k}^{i}\}_{k\in \mathbb{N}}$, and the triggering will not be continuous for infinite time to ensure the feasibility and avoid the Zeno behavior.

\section{Main Results \label{MainResults}}
In this section, we study the affine formation for linear MASs with available or unavailable full-state information.

\subsection{State-based Event-triggered Affine Formation
\label{State-based Event-triggered Affine Formation}}
Assume that the leaders' target states are time-invariant, that is, $\dot{p}^{*}_{l}=0$. When full-state information can be obtained, to trend the leader to converge to the target position, we design the ETC input as follows:
\begin{equation}\label{equation:2}
{u_i}(t) = K\hat{p}_i(t)+{v_i}, \quad i \in {V_l},
\end{equation}
where $\hat{p}_i(t)=p_i(t^{i}_{k})$ denotes each agent's state of the triggering time $t=t^{i}_{k}$, $K\in \mathbb{R}^{m\times d}$ is the gain matrix, and $v_i$ is the  compensation term satisfying that $(A+BK)p_i^{*}+Bv_i=0$.

For $i$-th follower, we propose the distributed ETC protocol:
\begin{align}\label{equation:3}
\left\{\begin{aligned}
&{u_i}(t)=K\hat{p}_i(t)+\hat{y}_i(t), \\
&{{\dot y}_i}(t) =  - {d_{i}(t)}\hat{\theta}_{i}(t), \quad {\dot d}_{i}(t) =  \hat{\theta}_{i}^{T}(t)\hat{\theta}_{i}(t)  ,\quad i \in {V_f},
\end{aligned} \right.
\end{align}
where $d_i(t)$ denotes the adaptive coupling
parameter with $d_i(0)>1$, $\hat{\theta}_{i}(t)=\sum_{j \in {V_{f}}} {w_{ij}({{\hat y}_{i}}(t) - {\hat y}_{j}(t))}+\sum_{j \in {V_{l}}} {w_{ij}({{\hat y}_{i}}(t) - v_{j})}$, $y_i(t)$ is devised to estimate the compensation of each follower and $\hat{y}_i(t)=y_{i}(t_{k}^{i})$.

Define the measurement errors of state and compensation of $i$-th agent as ${\epsilon_{p_{i}}}(t) = \hat{p}_i(t)-p_i(t)$ for $i \in {V}$ and ${\epsilon_{y_{i}}}(t)  = \hat{y}_i(t)-y_i(t)$ for $i \in {V_f}$. Motivated by \cite{Yan2021Robust}, the triggering time sequence $\{t^{i}_{k}\}$ of $i$-the agent satisfies the adaptive event-triggered strategy:
\begin{equation}\label{equation:5}
t_{k+1}^{i}=\mathrm{inf}\{t>t_{k}^{i}|f_{i}(t)>0\}, \quad i \in {V},
\end{equation}
with $t_{0}^{i} = 0$ and the triggering functions are defined by
\begin{align}
f_{i}(t)\!= & \gamma_{i}(t)\|{\epsilon_{p_{i}}}(t)\|^{2}\!-\!\|\hat{p}_i(t)-p_i^{*}\|^{2}\!-\!\mu_{1}e^{-\varpi_{1}t},  \, i \in {V_{l}},\label{equation:6.1}\\
f_{i}(t)\!= & \varphi_{i}(t)\|{\epsilon_{p_{i}}}(t)\|^{2}\!+\!\phi_{i}(t)d_{i}(t)\|{\epsilon_{y_{i}}}(t)\|^{2} \!-\!\|\hat{\xi}_{i}(t)\|^{2}\!\notag\\
&-\!\|\hat{\theta}_i(t)\|^{2}\!-\!\mu_{2}e^{-\varpi_{2}t}
,\quad i \in {V_f},\label{equation:6.2}
\end{align}
where $\hat{\xi}_{i}(t)=\sum_{j \in {\mathcal{N}_{i}}} {w_{ij}({{\hat p}_{i}}(t) - {\hat p}_{j}(t))}$ is the combined state measurement, $\mu_{1}$, $\mu_{2}$, $\varpi_{1}$, $\varpi_{2}$ are positive constants, and the adaptive parameters $\gamma_{i}(t)$, $\varphi_{i}(t)$, $\phi_{i}(t)$ satisfy the update laws:
\begin{align}
&\dot{\gamma}_{i}(t)= \|{\epsilon_{p_{i}}}(t)\|^{2}, \quad \gamma_{i}(0)>0, \quad i \in {V_{l}},\label{equation:7.1}\\
&\dot{\varphi}_{i}(t)= \|{\epsilon_{p_{i}}}(t)\|^{2}, \quad \varphi_{i}(0)>0, \quad i \in {V_f},\label{equation:7.2}\\
&\dot{\phi}_{i}(t)= d_{i}(t)\|{\epsilon_{y_{i}}}(t)\|^{2}, \quad d_{i}(0)>1, \quad i \in {V_f}.\label{equation:7.3}
\end{align}

The time variable $t$ is omitted in the analysis for brevity. Substituting the ETC protocols (\ref{equation:2}) and (\ref{equation:3}) into MASs (\ref{equation:1}), then we have MASs in the compact form:
\begin{align} \label{equation:8}
\left\{\begin{aligned}
{\dot{p}_l}\, \!=&\!\ [I_{n_{l}}\!\!\otimes \!(A\!+\!BK)]p_l\!+\!(I_{n_{l}}\!\!\otimes\! B)[(I_{n_{l}}\!\!\otimes\! K)\epsilon_{p_{l}}\!\! +\!v_l], \\
{\dot{p}_f}\!=&\!\ [I_{n_{f}}\!\otimes\! (A\!+\!BK)]p_f\!+\!(I_{n_{f}}\!\otimes \!BK)\epsilon_{p_{f}} \!\\&+\!(I_{n_{f}}\!\otimes\! B)y_{f}\!+\!(I_{n_{f}}\!\otimes\! B)\epsilon_{y_{f}}, \\
{{\dot y}_f}\!= &\! - (D\!\otimes\! I_{d})\theta_{f}\!-\!(D\Omega_{ff}\!\otimes \!I_{d})\epsilon_{y_{f}},\\
{{\dot d}_{f}}\!= &\! \ [\theta_{f}\!+\!(\Omega_{ff}\!\otimes \!I_{d})\epsilon_{y_{f}}]^{T}[\theta_{f}\!+\!(\Omega_{ff}\!\otimes \!I_{d})\epsilon_{y_{f}}],
\end{aligned}\right.
\end{align}
where $D={\rm diag}\{d^{T}_{f}\}={\rm diag}\{d_{n_{l}+1}, \dots, d_{n}\}$, $v_l=[v^{T}_{1},\dots,v^{T}_{n_{l}}]^{T}$, $y_f=[y^{T}_{n_{l}+1},\dots,y^{T}_{n}]^{T}$, $\epsilon_{p}=[\epsilon_{p_l}^{T},\epsilon_{p_f}^{T}]^{T}$, $\epsilon_{y_f}$, $\theta_{f}$ are defined similarly, and $\theta_{i}=\sum_{j \in {V_{f}}} {w_{ij}({y_{i}}-y_{j})}+\sum_{j \in {V_{l}}} {w_{ij}({y_{i}} - v_{j})}$.

Then, we give the theorem stating the solution for affine formation with state accessibility and communication limited.

\begin{theorem}\label{th1}
It is assumed that Assumptions \ref{a1}-\ref{a2} and the pair $(A,B)$ is stabilizable hold. The affine formation problem of MASs (\ref{equation:1}) under the control scheme (\ref{equation:2})-(\ref{equation:3}) with the adaptive event-triggered strategy (\ref{equation:5})-(\ref{equation:7.3}) can be tackled if $(A+BK)p_i^{*}+Bv_i=0$ for $i\in V_{l}$ and the control gain $K=-B^{T}P$, where $P$ is a positive definite matrix satisfies the algebraic Riccati equation $PA+A^{T}P-PBB^{T}P+R_{1}=0$ and $R_{1}$ is a positive definite matrix.
\end{theorem}
\begin{proof}
Define leaders' state errors $e_{p_{l}}=p_{l}-p^{*}_{l}$ and the formation errors ${\xi}_{f}=[\xi^{T}_{n_{l}+1},\dots,\xi^{T}_{n}]^{T}$, where $\xi_{i}=\sum_{j \in {\mathcal{N}_{i}}} {w_{ij}({p}_{i} - p_{j})}$. Then, by $(A+BK)p_i^{*}+Bv_i=0$ for $i\in V_{l}$, the error dynamics are
\begin{align}\left\{\begin{aligned} \label{equation:9}
{\dot{e}_{p_{l}}}\!\!=&\!\ [I_{n_{l}}\!\otimes \!(A\!+\!BK)]e_{p_{l}}\!+\!(I_{n_{l}}\!\otimes \! BK)\epsilon_{p_{l}}, \\
{\dot{\xi}_{f}}\,\!\!=&\!\ [I_{n_{f}}\!\otimes\! (A\!+\!BK)]\xi_f\!+\!(\Omega_{ff}\!\otimes \!BK)\epsilon_{p_{f}} \!\\&+\!(\Omega_{fl}\!\otimes \!BK)\epsilon_{p_{l}}\!+\!(I_{n_{f}}\!\otimes\! B)\theta_{f} \!+\!(\Omega_{ff}\!\otimes\! B)\epsilon_{y_{f}}, \\
{{\dot \theta}_f}\,\! \!= & \!-\! (\Omega_{ff}D\otimes I_{d})\theta_{f}\!-\!(\Omega_{ff}D\Omega_{ff}\otimes I_{d})\epsilon_{y_{f}},\\
{\dot d}_{f}\ \!\!\!= &\! \ [\theta_{f}\!+\!(\Omega_{ff}\!\otimes \!I_{d})\epsilon_{y_{f}}]^{T}[\theta_{f}\!+\!(\Omega_{ff}\!\otimes \!I_{d})\epsilon_{y_{f}}].
\end{aligned}\right.\end{align}

Establish the first part of Lyapunov function as
\begin{align*}
V_{1} \!=&\ \! \alpha_{1}\xi_{f}^{T}(I_{n_{f}}\!\otimes\! P)\xi_{f}\!+ \!\alpha_{2}\theta_{f}^{T}(\Omega_{ff}^{-1}\!\otimes \!I_{d})\theta_{f} \!+\!\alpha_{3}e_{p_{l}}^{T}(I_{n_{l}}\!\otimes\! P)e_{p_{l}},
\end{align*}
where $\alpha_{1}>\frac{3\alpha_{2}}{\lambda_{\min}(R_{1})}$, $\alpha_{3}>\frac{2\alpha_{2}}{\lambda_{\min}(R_{1})}$, and $\alpha_{2}>0$.

By computing the derivative of $V_{1}$ along the trajectory of (\ref{equation:9}) and defining $\Lambda=PBB^{T}P$, it obtains
\begin{align}\label{equation:11}
\dot V_{1} 
\leq&\ \alpha_{1}\Big\{\xi_{f}^{T}\Big[I_{n_{f}}\otimes (PA+A^{T}P)-I_{n_{f}}\otimes \Lambda\Big]\xi_{f}\notag\\
&+4\theta_{f}^{T}\theta_{f}+4\epsilon_{p_{f}}^{T}(\Omega_{ff}^{2}\otimes \Lambda)\epsilon_{p_{f}} \notag\\
&+4\epsilon_{p_{l}}^{T}(\Omega^{T}_{fl}\Omega_{fl}\otimes \Lambda)\epsilon_{p_{l}}+4\epsilon_{y_{f}}^{T}(\Omega_{ff}^{2}\otimes I_{d})\epsilon_{y_{f}} \Big\}\notag\\
&-2\alpha_{2}\left[\theta_{f}^{T}(D\otimes I_{d})\theta_{f} +\theta_{f}^{T}(D\Omega_{ff}\otimes I_{d})\epsilon_{y_{f}}\right]\notag\\
&+\alpha_{3}\left\{e_{p_{l}}^{T}[I_{n_{l}}\otimes (PA+A^{T}P)-I_{n_{l}}\otimes \Lambda]e_{p_{l}}\right.\notag\\
&\left.+\ \epsilon_{p_{l}}^{T}(I_{n_{l}}\otimes \Lambda)\epsilon_{p_{l}} \right\}.
\end{align}

The second part of the Lyapunov function is chosen as 
\begin{align}
V_{2} =\alpha_{2}\sum\limits_{i=n_{l}+1}^{n}{\frac{(d_{i}-\beta_{1})^{2}}{2}},\notag
\end{align}
where $\beta_{1}>\max\{1,\frac{8\alpha_{1}}{\alpha_{2}}-1\}$. 

Define $\hat{\theta}_{f}=[\hat{\theta}_{n_{l}+1}^{T},\dots,\hat{\theta}_{n}^{T}]^{T}$.
It follows from (\ref{equation:9}) that
\begin{align}\label{equation:13}
\dot V_{2} \leq &\  \alpha_{2}\left\{ \theta_{f}^{T}(D\otimes I_{d})\theta_{f}+\epsilon^{T}_{y_{f}}(\Omega_{ff}D\Omega_{ff}\otimes I_{d})\epsilon_{y_{f}} \right.\notag\\
&\left.+2\theta_{f}^{T}(D\Omega_{ff}\otimes I_{d})\epsilon_{y_{f}}-\frac{(\beta_{1}-1)}{2}\theta_{f}^{T}\theta_{f} \right.\notag\\
&\left.+(\beta_{1}-1)\epsilon^{T}_{y_{f}}(\Omega_{ff}\Omega_{ff}\otimes I_{d})\epsilon_{y_{f}}- \hat{\theta}_{f}^{T}\hat{\theta}_{f} \right\}.
\end{align}

Then, denote the third part of the Lyapunov function as
\begin{align*}
V_{3} = & \, \alpha_{2} \!\sum\limits_{i=n_{l}+1}^{n}{\!\frac{(\phi_{i}-\beta_{2})^{2}+(\varphi_{i}-\beta_{3})^{2}}{2}}+\alpha_{2}\sum\limits_{i=1}^{n_{l}}{\frac{(\gamma_{i}-\beta_{4})^{2}}{2}},
\end{align*}
where $\beta_{2}>(\frac{4\alpha_{1}}{\alpha_{2}}+\beta_{1})\|\Omega^{2}_{ff}\|$, $\beta_{3}>3\|\Omega^{2}_{ff}\|+\frac{4\alpha_{1}}{\alpha_{2}}\|\Omega^{2}_{ff}\otimes \Lambda\|$, and $\beta_{4}>2+\frac{\alpha_{3}}{\alpha_{2}} \|\Lambda\|+3\|\Omega_{fl}^{T}\Omega_{fl}\|+\frac{4\alpha_{1}}{\alpha_{2}}\|\Omega_{fl}^{T}\Omega_{fl}\otimes  \Lambda\|$. 

Substituting (\ref{equation:7.1})-(\ref{equation:7.3}) into the derivative of $V_{3}$ yields
\begin{align}\label{equation:15}
\dot V_{3} \leq &\  \alpha_{2}\Big\{ \sum\limits_{i=n_{l}+1}^{n}{\Big[\varphi_{i}\|{\epsilon_{p_{i}}}\|^{2}+\phi_{i}d_{i}\|{\epsilon_{y_{i}}}\|^{2}\Big]}+ \sum\limits_{i=1}^{n_{l}}{ \gamma_{i}\|{\epsilon_{p_{i}}}\|^{2}} \notag\\
&-\beta_{2}\epsilon^{T}_{y_{f}}(D\otimes I_{d})\epsilon_{y_{f}}-\beta_{3}\|\epsilon_{p_{f}}\|^{2}-\beta_{4}\|\epsilon_{p_{l}}\|^{2}\Big\}.
\end{align}

Since $\hat{\xi}_{f}=\xi_{f}+(\Omega_{ff}\otimes I_{d})\epsilon_{p_{f}}+(\Omega_{fl}\otimes I_{d})\epsilon_{p_{l}}$, where $\hat{\xi}_{f}=[\hat{\xi}^{T}_{n_{l}+1},\dots,\hat{\xi}^{T}_{n}]^{T}$, utilizing the event-triggered conditions (\ref{equation:6.1}) and (\ref{equation:6.2}), it follows from (\ref{equation:15}) that
\begin{align}\label{equation:16}
\dot V_{3} \leq &\  \alpha_{2}\Big\{ \hat{\theta}_{f}^{T}\hat{\theta}_{f}+3\xi_{f}^{T}\xi_{f}+3\epsilon_{p_{f}}^{T}(\Omega_{ff}^{2}\otimes I_{d})\epsilon_{p_{f}}+\Pi(t)\notag\\
&+3\epsilon_{p_{l}}^{T}(\Omega_{fl}^{T}\Omega_{fl}\otimes I_{d})\epsilon_{p_{l}} +2e_{p_{l}}^{T}e_{p_{l}}+2\epsilon_{p_{l}}^{T}\epsilon_{p_{l}}\notag\\
&-\beta_{2}\epsilon^{T}_{y_{f}}(D\otimes I_{d})\epsilon_{y_{f}}\!-\beta_{3}\|\epsilon_{p_{f}}\|^{2}-\beta_{4}\|\epsilon_{p_{l}}\|^{2}\},
\end{align}
where $\Pi(t)=(\mu_{1}n_{l}+\mu_{2}n_{f})e^{-\min{\{\varpi_{1},\varpi_{2}\}}t}$.

Given the candidate Lyapunov function $V=V_{1}+V_{2}+V_{3}$. 
Since $P$ satisfies the algebraic Riccati equation and $d_{i}(0)>1$, $\dot{d}_{i}\geq0$, combining (\ref{equation:11}), (\ref{equation:13}) and (\ref{equation:16}) provides
\begin{align*}
\dot V \leq &-\alpha_{1}\Big[\lambda_{\min}(R_{1})-\frac{3\alpha_{2}}{\alpha_{1}}\Big]\xi_{f}^{T}\xi_{f} \notag\\
&-\alpha_{2}\Big[\frac{(\beta_{1}\!+\!1)}{2}\!-\!\frac{4\alpha_{1}}{\alpha_{2}}\Big]\theta_{f}^{T}\theta_{f}\!-\!\alpha_{3}\Big[\lambda_{\min}(R_{1})\!-\!\frac{2\alpha_{2}}{\alpha_{3}}\Big]e_{p_{l}}^{T}e_{p_{l}} \notag\\
&-\alpha_{2}\Big[\beta_{2}-\|\Omega_{ff}^{2}\|-\Big(\frac{4\alpha_{1}}{\alpha_{2}}+\beta_{1}-1\Big)\|\Omega_{ff}^{2}\|\Big]\epsilon^{T}_{y_{f}}\epsilon_{y_{f}}\notag\\
&-\alpha_{2}\Big(\beta_{3}-3\|\Omega_{ff}^{2}\|-\frac{4\alpha_{1}}{\alpha_{2}}\|\Omega_{ff}^{2}\otimes \Lambda\|\Big) \epsilon_{p_{f}}^{T}\epsilon_{p_{f}}\notag\\
&-\alpha_{2}\Big(\beta_{4}-2-\frac{\alpha_{3}}{\alpha_{2}} \|\Lambda\|-3\|\Omega_{fl}^{T}\Omega_{fl}\|\notag\\
&-\frac{4\alpha_{1}}{\alpha_{2}}\|\Omega_{fl}^{T}\Omega_{fl}\!\otimes \! \Lambda\|\Big)\epsilon_{p_{l}}^{T}\epsilon_{p_{l}}\!+\!\Pi(t).
\end{align*}

Given the conditions satisfied by the parameters $\alpha_{i}$ ($i=1,2,3$) and $\beta_{i}$ ($i=1,2,3,4$), it can deduce that
\begin{align}\label{equation:19}
\dot V \leq &-\!\bar{\alpha}_{1}\xi_{f}^{T}\xi_{f}\! -\!\bar{\alpha}_{2}\theta_{f}^{T}\theta_{f}  \!-\!\bar{\alpha}_{3}e_{p_{l}}^{T}e_{p_{l}}\!+\!\Pi(t)
\leq \Pi(t).
\end{align}
where $\bar{\alpha}_{1}=\alpha_{1}\lambda_{\min}(R_{1})-3\alpha_{2}$,  $\bar{\alpha}_{2}=\frac{\alpha_{2}(\beta_{1}+1)}{2}-4\alpha_{1}$, $\bar{\alpha}_{3}=\alpha_{3}\lambda_{\min}(R_{1})-2\alpha_{2}$ are positive constants.

Integrating both sides of the second inequality of (\ref{equation:19}) yields
\begin{align}\label{equation:20}
{V}(t) \le& \ V(0)+\int_{0}^{t} \Pi(\tau)d\tau,
\end{align}
which means that $V$ is bounded, that is, $\xi_{i}$, $\theta_{i}$, $d_{i}$, $\phi_{i}$, $\varphi_{i}$ are bounded for $i\in V_{f}$ and $e_{p_{i}}$, $\gamma_{i}$ are bounded for $i\in V_{l}$.
Then, from the first inequality of (\ref{equation:19}), it follows that 
\begin{align}\label{equation:21}
\int_{0}^{t}{g}d\tau 
\le-V(t)+ V(0)+\int_{0}^{t} \Pi(\tau)d\tau,
\end{align}
which implies that $\int_{0}^{t}gd\tau=\int_{0}^{t}(\bar{\alpha}_{1}\xi_{f}^{T}\xi_{f}+\bar{\alpha}_{2}\theta_{f}^{T}\theta_{f}  +\bar{\alpha}_{3}e_{p_{l}}^{T}e_{p_{l}})d\tau$ is bounded. Since the derivative of $g$ is bounded, then $g$ is uniformly continuous. Given the existence of $\lim_{t\rightarrow\infty}\int_{0}^{t}{g}d\tau$, then it follows from the generalized Barbalat lemma \cite{Farkas2016Variations} that $\lim_{t\rightarrow\infty}{g}=0$.

Due to $\bar{\alpha}_{1}$, $\bar{\alpha}_{2}$ and $\bar{\alpha}_{3}$ are positive constants, it derives that $\lim_{t\rightarrow\infty}{\xi_{f}}=0$, $\lim_{t\rightarrow\infty}{\theta_{f}}=0$ and $\lim_{t\rightarrow\infty}{e_{p_{l}}}=0$. It indicates that MASs can converge to the affine formation and the adaptive parameters converge to some finite constants by equations (\ref{equation:2}), (\ref{equation:3}), (\ref{equation:7.1})-(\ref{equation:7.3}). So far, the proof is finalized.
\end{proof}

\begin{remark}
In the above proof, the communication information, such as the stress matrix, is adopted in the designed Lyapunov function, including the parameters $\beta_{1}$, $\beta_{2}$, $\beta_{3}$, $\beta_{4}$. However, the controller and the event-triggered strategy do not pertain to it, which is the advantage of the proposed adaptive control algorithm compared to the work in \cite{Ma2023Event}.
\end{remark}

\begin{remark} \label{r3}
Here are some notes on the condition $(A+BK)p_i^{*}+Bv_i=0$ for $i\in V_{l}$.

(1) If the control input matrix $B$ is of full row rank, then there exists a corresponding compensation term $v_i$ for any target formation, in which case the design does not constrain the mechanical structure of the system or the target formation. It means that the control inputs are redundant, which is applicable to some models with redundant actuators. This assumption has been made in the study \cite{Wang2024Neural}.

(2) On the other hand, if $B$ is of full column rank, the condition is equivalent to that there exists an invertible matrix $U=[U^{T}_{1},U^{T}_{2}]^{T}$ with $U_{1}\in \mathbb{R}^{m\times d}$ and $U_{2}\in \mathbb{R}^{(d-m)\times d}$ such that $U_{1}B=I_{m}$ and $U_{2}B=0$, then $p^{*}_{i}$ satisfies $U_{2}Ap^{*}_{i}=0$ and $v_{i}=-U_{1}(A+BK)p^{*}_{i}$. Then, restrictions need to be placed on the system's mechanical structure or the actuators' physical structure or target formation. The idea of the compensation term is also commonly found in the time-varying formation control, which is reasonable.
\end{remark}
\begin{remark}
The main difference between the time-varying and affine formations is that the time-varying formation requires each agent to know its desired position relative to the center or the leader's position function. In contrast, for arbitrary affine transformations of the nominal formation, only the leader knows its own desired position and the other followers can follow. In this paper, the follower's $y_{i}$ converges to the desired compensation term using an estimator approach.
\end{remark}

Next, it verifies the feasibility of the event-triggered strategy.
\begin{theorem}\label{th2}
Under Theorem \ref{th1} holds, the MASs are without Zeno behavior occurring.
\end{theorem}
\begin{proof}
First to prove each leader's time interval is strictly
more than zero, taking the time derivative for $\gamma_{i}\|{\epsilon_{p_{i}}}\|^{2}$ over the interval $[t_{k}^{i},t_{k+1}^{i})$ yields
\begin{align}\label{equation:22}
\frac{d(\gamma_{i}\|{\epsilon_{p_{i}}}\|^{2})}{dt}=&\ \dot{\gamma}_{i}\|{\epsilon_{p_{i}}}\|^{2}\!+2\gamma_{i}\epsilon^{T}_{p_{i}}(A\epsilon_{p_{i}}\!-(A+\!BK)\hat{p}_{i}\!-Bv_{i})\notag\\
\leq&\ \Big(\frac{\dot{\gamma}_{i}}{\gamma_{i}}\!+\!2\|A\|\!+\!2\Big)\gamma_{i}\|{\epsilon_{p_{i}}}\|^{2}\!+\!\gamma_{i}\|(A\!+\!BK)\hat{p}_{i}\|^{2}\notag\\
&+\gamma_{i}\|Bv_{i}\|^{2}\notag\\
\leq&\ \bar{A}_{1i}\gamma_{i}\|{\epsilon_{p_{i}}}\|^{2}+\Delta_{1i},
\end{align}
where $\bar{A}_{1i}=\max_{t\in[t_{k}^{i},t_{k+1}^{i})}\big\{\frac{\dot{\gamma}_{i}}{\gamma_{i}}\!+2\|A\|\!+2\big\}$ and $\Delta_{1i}=\max_{t\in[t_{k}^{i},t_{k+1}^{i})}\{\gamma_{i}\|(A+\!BK)\hat{p}_{i}\|^{2}+\gamma_{i}\|Bv_{i}\|^{2}\}$.

Upon triggering the next event, the triggering function (\ref{equation:6.1}) should be set to zero and integrating  (\ref{equation:22}) can obtain that
\begin{align*}
\|p_i(t_{k+1}^{i})-p_i^{*}\|^{2}+&\ \mu_{1}e^{-\varpi_{1}t_{k+1}^{i}}
\leq\frac{\Delta_{1i}}{\bar{A}_{1i}}(e^{\bar{A}_{1i}(t_{k+1}^{i}-t_{k}^{i})}-1).
\end{align*}
Following $\|p_i(t_{k+1}^{i})-p_i^{*}\|^{2}\geq0$, it is derived that
\begin{align}\label{equation:25}
t_{k+1}^{i}-t_{k}^{i}
\geq&\ \frac{1}{\bar{A}_{1i}}\ln(\frac{\mu_{1}\bar{A}_{1i}}{\Delta_{1i}}e^{-\varpi_{1}t_{k+1}^{i}}+1)>0.
\end{align}

Then, a similar derivation for each follower yields
\begin{align*}
\frac{d(\varphi_{i}\|{\epsilon_{p_{i}}}\!\|^{2}\!+\!\!\phi_{i}d_{i}\|{\epsilon_{y_{i}}}\!\|^{2})}{dt}\leq \bar{A}_{2i}(\varphi_{i}\|{\epsilon_{p_{i}}}\!\|^{2}\!+\!\phi_{i}d_{i}\|{\epsilon_{y_{i}}}\!\|^{2})\!+\!\Delta_{2i},
\end{align*}
where $\bar{A}_{2i}=\max_{t\in[t_{k}^{i},~t_{k+1}^{i})}\big\{\frac{\dot{\varphi}_{i}}{\varphi_{i}}+2\|A\|+2,\dot{\phi}_{i}+\dot{d}_{i}+1\big\}$ and $\Delta_{2i}=\max_{t\in[t_{k}^{i},t_{k+1}^{i})}\{\phi_{i}d^{3}_{i}\|\hat{\theta}_i\|^{2}+\varphi_{i}\|(A+BK)\hat{p}_{i}\|^{2}+\varphi_{i}\|B\hat{y}_{i}\|^{2}\}$,
and then we obtain that
\begin{align}\label{equation:29}
t_{k+1}^{i}-t_{k}^{i}
\geq&\ \frac{1}{\bar{A}_{2i}}\ln(\frac{\mu_{2}\bar{A}_{2i}}{\Delta_{2i}}e^{-\varpi_{2}t_{k+1}^{i}}+1)>0.
\end{align}

Thus, the inter-event time interval is strictly more than zero by (\ref{equation:25}) and (\ref{equation:29}). The proof is completed.
\end{proof}

\subsection{Output-based Event-triggered Affine Formation \label{Output-based Event-triggered Affine Formation}}
In real-world scenarios, difficulty in measuring information renders the unavailability and transmission of the agent's state. Hence, based on the available output information, the observer-based distributed ETC protocol is proposed:
\begin{align}\label{equation:40}
\left\{\begin{aligned}
&{{\dot z}_i}(t) \!= Az_i(t)\!+\!Bu_{i}(t)\!+\!F(C{z_i}(t)\!\!-{q_i}(t)), \quad i \in {V},\\
&{u_i}(t)\! = K\hat{z}_i(t)\!+\!{v_i}, \quad i \in {V_l},\\
&{u_i}(t)\!=K\hat{z}_i(t)\!+\!\hat{y}_i(t),\quad i \in {V_f}, \\
&{{\dot y}_i}(t)\! =  - {d_{i}(t)}\hat{\theta}_{i}(t),\quad  {\dot d}_{i}(t)\! =  \hat{\theta}_{i}^{T}(t)\hat{\theta}_{i}(t)  ,\quad i \in {V_f},
\end{aligned} \right.
\end{align}
where ${z_i}(t)$ is adopted to observe the state of each agent, $F\in\mathbb{R}^{d\times r}$ denotes an observer gain matrix such that $A+FC$ is Hurwitz, $\hat{z}_i(t)=z_i(t^{i}_{k})$, and $y_{i}(t)$, $\hat{\theta}_{i}(t)$, $d_{i}(t)$ for $i\in V_{f}$, $v_i$ for $i\in V_{l}$ are designed in Section \ref{State-based Event-triggered Affine Formation}.

By defining the measurement errors of observed state and estimated compensation of $i$-th agent as ${\epsilon_{z_{i}}}(t) = \hat{z}_i(t)-z_i(t)$ for $i \in {V}$, and ${\epsilon_{y_{i}}}(t)  = \hat{y}_i(t)-y_i(t)$ for $i \in {V_f}$, then the triggering function for leaders and followers are redefined by
\begin{align}
f_{i}(t)\!=&\gamma_{i}(t)\|{\epsilon_{z_{i}}}(t)\|^{2}\!-\!\|\hat{z}_i(t)\!-\!p_i^{*}\|^{2}\!-\!\mu_{3}e^{-\varpi_{3}t}, \,i \in {V_{l}}, \label{equation:42.1} \\
f_{i}(t)\!=&\varphi_{i}(t)\|{\epsilon_{z_{i}}}(t)\|^{2}\!+\!\phi_{i}(t)d_{i}(t)\|{\epsilon_{y_{i}}}(t)\|^{2} \!-\!\|\hat{\eta}_{i}(t)\|^{2}\notag\\
&-\|\hat{\theta}_i(t)\|^{2}\!-\!\mu_{4}e^{-\varpi_{4}t}
,\quad i \in {V_f},\label{equation:42.2}
\end{align}
where $\hat{\eta}_{i}(t)=\sum_{j \in {\mathcal{N}_{i}}} {w_{ij}({{\hat z}_{i}}(t) - {\hat z}_{j}(t))}$ is the combined observed-state measurement, $\mu_{3}$, $\mu_{4}$, $\varpi_{3}$, $\varpi_{4}>0$, and adaptive parameters $\gamma_{i}(t)$, $\varphi_{i}(t)$, $\phi_{i}(t)$ satisfy 
\begin{align}
&\dot{\gamma}_{i}(t)= \|{\epsilon_{z_{i}}}(t)\|^{2}, \quad \gamma_{i}(0)>0, \quad i \in {V_{l}},\label{equation:43.1}\\
&\dot{\varphi}_{i}(t)= \|{\epsilon_{z_{i}}}(t)\|^{2}, \quad \varphi_{i}(0)>0, \quad i \in {V_f},\label{equation:43.2}\\
&\dot{\phi}_{i}(t)= d_{i}(t)\|{\epsilon_{y_{i}}}(t)\|^{2}, \quad d_{i}(0)>1, \quad i \in {V_f},\label{equation:43.3}
\end{align}

The time variable $t$ is omitted. By taking the control protocol (\ref{equation:40}) into MASs (\ref{equation:1}), then we obtain the closed-loop system in the compact form:
\begin{align} \label{equation:44}
\left\{\begin{aligned}
{\dot{p}_l}\,\!\!=&\ \!(I_{n_{l}}\!\otimes \!A)p_l\!+\!(I_{n_{l}}\!\otimes\! BK)(z_{l} \! +\!\epsilon_{z_{l}})\!+\!(I_{n_{l}}\!\otimes\! B)v_l, \\
{\dot{z}_l}\,\!=&\ \![I_{n_{l}}\!\otimes\! (A+BK)]z_l\!+\!(I_{n_{l}}\!\otimes \!BK)\epsilon_{z_{l}} \\
&+(I_{n_{l}}\otimes B)v_l+(I_{n_{l}}\otimes FC)\delta_{l}, \\
{\dot{p}_f}\!\!=&\ \!(I_{n_{f}}\!\otimes \!A)p_f\!+\!(I_{n_{f}}\!\otimes\! BK)z_{f} \!+\!(I_{n_{f}}\!\otimes \!BK)\epsilon_{z_{f}}\\
&+(I_{n_{f}}\otimes B)y_{f}+(I_{n_{f}}\otimes B)\epsilon_{y_{f}}, \\
{\dot{z}_f}\!\!=&\ \! [I_{n_{f}}\!\otimes \!(A+BK)]z_f\!+\!(I_{n_{f}}\!\otimes\! BK)\epsilon_{z_{f}}\\
&+(I_{n_{f}}\!\otimes\! B)y_{f}\!+\!(I_{n_{f}}\!\otimes\! B)\epsilon_{y_{f}}\!+\!(I_{n_{f}}\!\otimes\! FC)\delta_{f}, \\
{{\dot y}_f}\!\!= &\! - (D\!\otimes\! I_{d})\theta_{f}-(D\Omega_{ff}\!\otimes \!I_{d})\epsilon_{y_{f}},\\
{\dot d_{f}}\ \! \!\!= & \!\ [\theta_{f}\!+\!(\Omega_{ff}\!\otimes \!I_{d})\epsilon_{y_{f}}]^{T}[\theta_{f}\!+\!(\Omega_{ff}\!\otimes \!I_{d})\epsilon_{y_{f}}],
\end{aligned}\right.\end{align}
where the linked observed states $z=[z_l^{T},z_f^{T}]^{T}=[z_1^{T},\dots,z_{n_{l}}^{T},z_{n_{l}+1}^{T},\dots,z_{n}^{T}]^{T}$, the measurement errors $\epsilon_{z}=[\epsilon_{z_l}^{T},\epsilon_{z_f}^{T}]^{T}$, the observed errors $\delta=[\delta^{T}_{l},\delta^{T}_{f}]^{T}=z-p$, and $D$, $d_{f}$, $y_f$, $\epsilon_{p}$, $\epsilon_{y_f}$, $\theta_{f}$, $v_l$ are defined in Section \ref{State-based Event-triggered Affine Formation}.

Next, the results regarding affine formation with state unavailability and limited communication are shown below.

\begin{theorem}\label{th3}
Suppose that Assumptions \ref{a1}-\ref{a2} as well as the stabilizable and detectable triple $(A,B,C)$ hold. The affine formation problem of MASs (\ref{equation:1}) under the control scheme (\ref{equation:40}) with the adaptive event-triggered strategy (\ref{equation:5}), (\ref{equation:42.1})-(\ref{equation:43.3}) can be solved if $(A+BK)p_i^{*}+Bv_i=0$ for $i\in V_{l}$ and the control gain matrices $K=-B^{T}P$ and  $F=-QC^{T}$, where $P$ and $Q$ are positive definite matrices satisfy the algebraic Riccati equations $PA+A^{T}P-PBB^{T}P+R_{1}=0$ and $QA^{T}+AQ-QC^{T}CQ+R_{2}=0$, and $R_{1}$, $R_{2}$ are positive definite matrices with $\lambda_{\min}(R_{1})> 1/2$.
\end{theorem}
\begin{proof}
Define the formation errors ${\eta}_{f}=[\eta^{T}_{n_{l}+1},\dots,\eta^{T}_{n}]^{T}$, where $\eta_{i}=\sum_{j \in {\mathcal{N}_{i}}} {w_{ij}({z}_{i}- z_{j})}$. Denote $e_{z_{l}}=z_{l}-p^{*}_{l}$ as the differences between observed and desired states for leaders.

Combing the condition $(A+BK)p_i^{*}+Bv_i=0$ for $i\in V_{l}$ yields the error dynamics as
\begin{align}\left\{\begin{aligned} \label{equation:45}
\dot{\delta}\ \ =&\ [I_{n}\otimes (A+FC)]\delta,\\
{\dot{e}_{z_{l}}}=&\ [I_{n_{l}}\!\otimes\! (A\!+\!BK)]e_{z_{l}}\!+\!(I_{n_{l}}\!\otimes\! BK)\epsilon_{z_{l}} \\
&+(I_{n_{l}}\otimes FC)\delta_{l}, \\
{\dot{\eta}_{f}}\,=&\ [I_{n_{f}}\!\otimes\! (A\!+\!BK)]\eta_f\!+\!(\Omega_{ff}\!\otimes\! BK)\epsilon_{z_{f}} \\&+\!(\Omega_{fl}\!\otimes\! BK)\epsilon_{z_{l}}
+(I_{n_{f}}\otimes B)\theta_{f}\\&+(\Omega_{ff}\otimes B)\epsilon_{y_{f}}+(\Omega_{ff}\otimes FC)\delta_{f} \\
&+(\Omega_{fl}\otimes FC)\delta_{l}, \\
{{\dot \theta}_f}\, = & - [\Omega_{ff}D\otimes I_{d}]\theta_{f}-[\Omega_{ff}D\Omega_{ff}\otimes I_{d}]\epsilon_{y_{f}},\\
{\dot{d}_{f}}\,= & \ \theta_{f}^{T}\theta_{f}+\epsilon^{T}_{y_{f}}(\Omega_{ff}\Omega_{ff}\otimes I_{d})\epsilon_{y_{f}}\\
&+2\theta_{f}^{T}(\Omega_{ff}\otimes I_{d})\epsilon_{y_{f}}.
\end{aligned}\right.\end{align}

Establish the candidate Lyapunov function as
\begin{align}\label{equation:46}
V=&\ V_{1}+V_{2}+V_{3}+V_{4},\\
V_{1} =&\  \alpha_{1}\eta_{f}^{T}(I_{n_{f}}\otimes P)\eta_{f}+ \alpha_{2}\theta_{f}^{T}(\Omega_{ff}^{-1}\otimes I_{d})\theta_{f} \notag\\
&+\alpha_{3}e_{z_{l}}^{T}(I_{n_{l}}\otimes P)e_{z_{l}},\notag\\
V_{2} = & \ \alpha_{2} \sum\limits_{i=n_{l}+1}^{n}{\frac{(d_{i}-\beta_{1})^{2}}{2}},\notag\\
V_{3} = & \ \alpha_{2}\sum\limits_{i=n_{l}+1}^{n}{\frac{(\phi_{i}-\beta_{2})^{2}}{2}}+\alpha_{2}\sum\limits_{i=n_{l}+1}^{n}{\frac{(\varphi_{i}-\beta_{3})^{2}}{2}}\notag\\
&+\alpha_{2}\sum\limits_{i=1}^{n_{l}}{\frac{(\gamma_{i}-\beta_{4})^{2}}{2}},\notag\\
V_{4} = & \ \alpha_{4}\delta(t)^{T}(I_{n}\otimes Q^{-1})\delta(t),\notag
\end{align}
where $\alpha_{1}>\frac{6\alpha_{2}}{2\lambda_{\min}(R_{1})-1}$, $\alpha_{2}>0$, $\alpha_{3}>\frac{4\alpha_{2}}{2\lambda_{\min}(R_{1})-1}$, $\alpha_{4}>\max\big\{\frac{4\alpha_1\|\Omega_{ff}^{2}\|\|PFC\|^{2}}{\lambda_{\min}(S)},
\frac{2\alpha_{3}\|PFC\|^{2}+4\alpha_1\|\Omega_{fl}^{T}\Omega_{fl}\|\|PFC\|^{2}}{\lambda_{\min}(S)}\big\}$, $\beta_{4}>2+\frac{\alpha_{3}}{\alpha_{2}} \|\Lambda\|+3\|\Omega_{fl}^{T}\Omega_{fl}\|+\frac{4\alpha_{1}}{\alpha_{2}}\|\Omega_{fl}^{T}\Omega_{fl}\otimes \Lambda\|$, $\beta_{3}>3\|\Omega_{ff}^{2}\|+\frac{4\alpha_{1}}{\alpha_{2}}\|\Omega_{ff}^{2}\otimes \Lambda\|$, $\beta_{2}>(\frac{4\alpha_{1}}{\alpha_{2}}+\beta_{1})\|\Omega^{2}_{ff}\|$, $\beta_{1}>\frac{8\alpha_{1}}{\alpha_{2}}-1$, and $S$ is a positive definite matrix to be designed later.

First, computing the derivative of $V_{1}$ along the trajectory of (\ref{equation:45}), it obtains
\begin{align}\label{equation:47}
\dot V_{1} =&\ \alpha_{1}\left\{\eta_{f}^{T}[I_{n_{f}}\otimes (PA+A^{T}P)-2I_{n_{f}}\otimes \Lambda]\eta_{f}\right.\notag\\
&\left.+2\eta_{f}^{T}(I_{n_{f}}\otimes PB)\theta_{f}-2\eta_{f}^{T}(\Omega_{ff}\otimes \Lambda)\epsilon_{z_{f}}\right.\notag\\
&\left.-2\eta_{f}^{T}(\Omega_{fl}\otimes \Lambda)\epsilon_{z_{l}}+2\eta_{f}^{T}(\Omega_{ff}\otimes PB)\epsilon_{y_{f}}\right.\notag\\
&\left.+2\eta_{f}^{T}(\Omega_{ff}\otimes PFC)\delta_{f}+2\eta_{f}^{T}(\Omega_{fl}\otimes PFC)\delta_{l}\right\}\notag\\
&-2\alpha_{2}\left[\theta_{f}^{T}(D\otimes I_{d})\theta_{f} +\theta_{f}^{T}(D\Omega_{ff}\otimes I_{d})\epsilon_{y_{f}}\right]\notag\\
&+\alpha_{3}\left\{e_{z_{l}}^{T}[I_{n_{l}}\otimes (PA+A^{T}P)-2I_{n_{l}}\otimes \Lambda]e_{z_{l}}\right.\notag\\
&\left.-2e_{z_{l}}^{T}(I_{n_{l}}\otimes \Lambda)\epsilon_{z_{l}} +2e_{z_{l}}^{T}(I_{n_{l}}\otimes PFC)\delta_{l}\right\}\notag\\
\leq&\ \alpha_{1}\Big\{\eta_{f}^{T}\Big[I_{n_{f}}\otimes (PA\!+\!A^{T}P)\!-\!I_{n_{f}}\otimes \Lambda\!+\!\frac{1}{2}I_{dn_{f}}\Big]\eta_{f}\notag\\
&+4\theta_{f}^{T}\theta_{f}+4\epsilon_{z_{f}}^{T}(\Omega_{ff}^{2}\otimes \Lambda)\epsilon_{z_{f}} +4\epsilon_{z_{l}}^{T}(\Omega_{fl}^{T}\Omega_{fl}\otimes \Lambda)\epsilon_{z_{l}} \notag\\
&+4\epsilon_{y_{f}}^{T}(\Omega_{ff}^{2}\otimes I_{d})\epsilon_{y_{f}}\!+\!4\delta_{f}^{T}(\Omega_{ff}^{2}\otimes C^{T}F^{T}PPFC)\delta_{f}\notag\\
&+4\delta_{l}^{T}(\Omega_{fl}^{T}\Omega_{fl}\otimes C^{T}F^{T}PPFC)\delta_{l} \Big\}\notag\\
&-2\alpha_{2}\left[\theta_{f}^{T}(D\otimes I_{d})\theta_{f} +\theta_{f}^{T}(D\Omega_{ff}\otimes I_{d})\epsilon_{y_{f}}\right]\notag\\
&+\alpha_{3}\Big\{e_{z_{l}}^{T}\Big[I_{n_{l}}\otimes (PA\!+\!A^{T}P)\!-\!I_{n_{l}}\otimes \Lambda\!+\!\frac{1}{2}I_{dn_{l}}\Big]e_{z_{l}}\notag\\
&+\epsilon_{z_{l}}^{T}(I_{n_{l}}\!\otimes \!\Lambda)\epsilon_{z_{l}} \!+\!2\delta_{l}^{T}(I_{n_{l}}\!\otimes \!C^{T}F^{T}PPFC)\delta_{l} \Big\},
\end{align}
where $\Lambda=PBB^{T}P$.

And the derivative of $V_{2}$ along the trajectory of (\ref{equation:45}) is
\begin{align}\label{equation:48}
\dot V_{2} \leq &\  \alpha_{2}\left\{ \theta_{f}^{T}(D\otimes I_{d})\theta_{f}+\epsilon^{T}_{y_{f}}(\Omega_{ff}D\Omega_{ff}\otimes I_{d})\epsilon_{y_{f}} \right.\notag\\
&\left.+2\theta_{f}^{T}(D\Omega_{ff}\otimes I_{d})\epsilon_{y_{f}}-\frac{(\beta_{1}-1)}{2}\theta_{f}^{T}\theta_{f} \right.\notag\\
&\left.+(\beta_{1}-1)\epsilon^{T}_{y_{f}}(\Omega_{ff}\Omega_{ff}\otimes I_{d})\epsilon_{y_{f}}- \hat{\theta}_{f}^{T}\hat{\theta}_{f} \right\}.
\end{align}

Then, taking equations (\ref{equation:43.1})-(\ref{equation:43.3}) into the derivative of $V_{3}$ and subject to triggering functions (\ref{equation:42.1}) and (\ref{equation:42.2}), one has
\begin{align}\label{equation:49}
\dot V_{3} \leq &\  \alpha_{2}\Big\{ \hat{\theta}_{f}^{T}\hat{\theta}_{f}+3\eta_{f}^{T}\eta_{f}+3\epsilon_{z_{f}}^{T}(\Omega_{ff}^{2}\otimes I_{d})\epsilon_{z_{f}}\notag\\
&+3\epsilon_{z_{l}}^{T}(\Omega_{fl}^{T}\Omega_{fl}\otimes I_{d})\epsilon_{z_{l}} +2e_{z_{l}}^{T}e_{z_{l}}+2\epsilon_{z_{l}}^{T}\epsilon_{z_{l}}\notag\\
&-\beta_{2}\epsilon^{T}_{y_{f}}(D\otimes I_{d})\epsilon_{y_{f}}-\beta_{3}\|\epsilon_{z_{f}}\|^{2}-\beta_{4}\|\epsilon_{z_{l}}\|^{2}\notag\\
&+(\mu_{3}n_{l}+\mu_{4}n_{f})e^{-\bar{\varpi}_{\min} t}\Big\},
\end{align}
where $\bar{\varpi}_{\min} =\min{\{\varpi_{3},\varpi_{4}\}}$.

Since $A + FC$ is Hurwitz, there exists a positive definite matrix $Q^{-1}$ satisfying $Q^{-1}(A+FC)+(A+FC)^{T}Q^{-1}=-S$, where $S$ is the arbitrary positive definite matrix. 

Then, differentiating $V_{4}$ along the trajectory of (\ref{equation:45}) yields
\begin{align}\label{equation:50}
\dot V_{4} = &\  \alpha_{4} \delta^{T}\Big\{I_{n}\otimes \big[Q^{-1}(A+FC)+(A+FC)^{T}Q^{-1}\big]\Big\}\delta \notag \\
 = & -\alpha_{4} \delta^{T}(I_{n}\otimes S)\delta.
\end{align}
Let $S=C^{T}C+Q^{-1}Q^{-1}$ and $F=-QC^{T}$, then the equation $Q^{-1}(A+FC)+(A+FC)^{T}Q^{-1}=-S$ can be converted into $QA^{T}+AQ-QC^{T}CQ+I_{d}=0$, which means that solving the above equation will yield the matrix $Q$.

Furthermore, following equations (\ref{equation:46})-(\ref{equation:50}), one obtains the derivative of $V$ as
\begin{align*}
\dot V \leq &\  \alpha_{1}\eta_{f}^{T}\Big[I_{n_{f}}\otimes (PA+A^{T}P- \Lambda)+\Big(\frac{1}{2}+\frac{3\alpha_{2}}{\alpha_{1}}\Big)I_{dn_{f}}\Big]\eta_{f} \notag\\
&-\alpha_{2}\theta_{f}^{T}\Big\{D\otimes I_{d}+\Big[\frac{(\beta_{1}-1)}{2}-\frac{4\alpha_{1}}{\alpha_{2}}\Big]I_{dn_{f}}\Big\}\theta_{f}  \notag\\
&+\alpha_{3}e_{z_{l}}^{T}\Big[I_{n_{f}}\otimes (PA\!+\!A^{T}P\!-\!\Lambda)\!+\!\Big(\frac{1}{2}\!+\!\frac{2\alpha_{2}}{\alpha_{3}}\Big)I_{dn_{l}}\Big]e_{z_{l}}\notag\\
&-\alpha_{2}\epsilon^{T}_{y_{f}}\Big\{\big[\beta_{2}D-\Big(\frac{4\alpha_{1}}{\alpha_{2}}+\beta_{1}-1\Big)\Omega_{ff}^{2}\notag\\
&-\Omega^{T}_{ff}D\Omega_{ff}\big]\otimes I_{d}\Big\}\epsilon_{y_{f}}\!-\!\alpha_{2} \epsilon_{z_{f}}^{T}\Big(\beta_{3}I_{dn_{f}}\!-\!3\Omega_{ff}^{2}\otimes I_{d}\notag\\
&-\frac{4\alpha_{1}}{\alpha_{2}}\Omega_{ff}^{2}\otimes \Lambda\Big)\epsilon_{z_{f}}\!-\!\alpha_{2}\epsilon_{z_{l}}^{T}\Big[(\beta_{4}\!-\!2)I_{dn_{l}}\!-\!\frac{\alpha_{3}}{\alpha_{2}}I_{n_{l}}\otimes \Lambda\notag\\
&-3\Omega_{fl}^{T}\Omega_{fl}\otimes I_{d}-\frac{4\alpha_{1}}{\alpha_{2}}\Omega_{fl}^{T}\Omega_{fl}\otimes \Lambda\Big]\epsilon_{z_{l}}\notag\\
&-\alpha_{4} \delta_{f}^{T}\Big(I_{n}\otimes S-\frac{4\alpha_{1}}{\alpha_{4}}\Omega_{ff}^{2}\otimes C^{T}F^{T}PPFC\Big)\delta_{f}\notag\\
&-\alpha_{4} \delta_{l}^{T}\Big[I_{n}\otimes \Big(S-\frac{2\alpha_{3}}{\alpha_{4}}C^{T}F^{T}PPFC\Big)\notag\\
&-\frac{4\alpha_{1}}{\alpha_{4}}\Omega_{fl}^{T}\Omega_{fl}\otimes C^{T}F^{T}PPFC\Big]\delta_{l} \notag\\
&+(\mu_{3}n_{l}+\mu_{4}n_{f})e^{-\bar{\varpi}_{\min} t}.
\end{align*}

Under the condition satisfied by $P$, one has
\begin{align*}
\dot V \leq &-\alpha_{1}\Big(\lambda_{\min}(R_{1})-\frac{1}{2}-\frac{3\alpha_{2}}{\alpha_{1}}\Big)\eta_{f}^{T}\eta_{f}\notag\\
&-\alpha_{2}\Big[\frac{(\beta_{1}+1)}{2}-\frac{4\alpha_{1}}{\alpha_{2}}\Big]\theta_{f}^{T}\theta_{f} \notag\\
&-\alpha_{3}\Big(\lambda_{\min}(R_{1})-\frac{1}{2}-\frac{2\alpha_{2}}{\alpha_{3}}\Big)e_{z_{l}}^{T}e_{z_{l}}\notag\\
&-\alpha_{2}\Big[\beta_{2}-\|\Omega_{ff}^{2}\|-\Big(\frac{4\alpha_{1}}{\alpha_{2}}+\beta_{1}-1\Big)\|\Omega_{ff}^{2}\|\Big]\epsilon^{T}_{y_{f}}\epsilon_{y_{f}}\notag\\
&-\alpha_{2}\Big(\beta_{3}-3\|\Omega_{ff}^{2}\|-\frac{4\alpha_{1}}{\alpha_{2}}\|\Omega_{ff}^{2}\otimes \Lambda\|\Big) \epsilon_{z_{f}}^{T}\epsilon_{z_{f}}\notag\\
&-\alpha_{2}\Big(\beta_{4}-2-\frac{\alpha_{3}}{\alpha_{2}} \|\Lambda\|-3\|\Omega_{fl}^{T}\Omega_{fl}\|\notag\\
&-\frac{4\alpha_{1}}{\alpha_{2}}\|\Omega_{fl}^{T}\Omega_{fl}\otimes \Lambda\|\Big)\epsilon_{z_{l}}^{T}\epsilon_{z_{l}}\notag\\
&-\alpha_{4}\Big(\lambda_{\min}(S)-\frac{4\alpha_{1}}{\alpha_{4}}\|\Omega_{ff}^{2}\|\|PFC\|^{2}\Big)  \delta_{f}^{T}\delta_{f}\notag\\
&-\alpha_{4}\Big(\lambda_{\min}(S)-\frac{2\alpha_{3}}{\alpha_{4}}\|PFC\|^{2}\notag\\
&-\frac{4\alpha_{1}}{\alpha_{4}}\|\Omega_{fl}^{T}\Omega_{fl}\|\|PFC\|^{2}\Big) \delta_{l}^{T}\delta_{l} \notag\\
&+(\mu_{3}n_{l}+\mu_{4}n_{f})e^{-\varpi_{\min} t}.
\end{align*}

Furthermore, given the conditions satisfied by the parameters $\alpha_{i}$ ($i=1,2,3,4$) and $\beta_{i}$ ($i=1,2,3,4$), it can infer that
\begin{align}\label{equation:53}
\dot V \leq &-\bar{\alpha}_{1}\eta_{f}^{T}\eta_{f} -\bar{\alpha}_{2}\theta_{f}^{T}\theta_{f}  -\bar{\alpha}_{3}e_{z_{l}}^{T}e_{z_{l}}-\bar{\alpha}_{4}\delta_{f}^{T}\delta_{f}-\tilde{\alpha}_{4}\delta_{l}^{T}\delta_{l}\notag\\
&\ +(\mu_{3}n_{l}+\mu_{4}n_{f})e^{-\varpi_{\min} t} \notag\\
\leq&\  (\mu_{3}n_{l}+\mu_{4}n_{f})e^{-\varpi_{\min} t},
\end{align}
where $\bar{\alpha}_{1}=\alpha_{1}\lambda_{\min}(R_{1})-\frac{\alpha_{1}}{2}-3\alpha_{2}$,  $\bar{\alpha}_{2}=\frac{\alpha_{2}(\beta_{1}+1)}{2}-4\alpha_{1}$, $\bar{\alpha}_{3}=\alpha_{3}\lambda_{\min}(R_{1})-\frac{\alpha_{3}}{2}-2\alpha_{2}$, $\bar{\alpha}_{4}=\alpha_{4}\lambda_{\min}(S)-4\alpha_1\|\Omega_{ff}^{2}\|\|PFC\|^{2}$ and $\tilde{\alpha}_{4}=\alpha_{4}\lambda_{\min}(S)-2\alpha_{3}\|PFC\|^{2}-4\alpha_1\|\Omega_{fl}^{T}\Omega_{fl}\|\|PFC\|^{2}$ are positive constants.

Referring to (\ref{equation:20}) and (\ref{equation:21}) in the proof of Theorem \ref{th1}, similarly we can obtain from (\ref{equation:53}) that $\int_{0}^{t}{g(\tau)}d\tau$ is bounded, where $g=\bar{\alpha}_{1}\eta_{f}^{T}\eta_{f} +\bar{\alpha}_{2}\theta_{f}^{T}\theta_{f}  +\bar{\alpha}_{3}e_{z_{l}}^{T}e_{z_{l}} +\bar{\alpha}_{4}\delta_{f}^{T}\delta_{f} +\tilde{\alpha}_{4}\delta_{l}^{T}\delta_{l}$. Since the derivative of function $g(t)$ is bounded, then $g(t)$ is uniformly continuous. Combined with $\lim_{t\rightarrow\infty}\int_{0}^{t}{g(\tau)}d\tau$ exists, then $\lim_{t\rightarrow\infty}{g(t)}=0$ follows from the generalized Barbalat lemma. Since $\bar{\alpha}_{i}$ ($i=1,2,3,4$) and $\tilde{\alpha}_{4}$ are positive constants, it deduced that $\lim_{t\rightarrow\infty}{\eta_{f}}=0$, $\lim_{t\rightarrow\infty}{\theta_{f}}=0$ and $\lim_{t\rightarrow\infty}{e_{z_{l}}}=0$ and $\lim_{t\rightarrow\infty}{\delta}=0$. It implies that MASs can converge to the affine formation and the adaptive parameters $d_{i}$, $\phi_{i}$, $\varphi_{i}$ ($i\in V_{f}$), $\gamma_{i}$ ($i\in V_{l}$) converge to some finite constants by equations (\ref{equation:40}), (\ref{equation:43.1})-(\ref{equation:43.3}). This concludes the proof.
\end{proof}

The following theorem clarifies the practicality of the adaptive event-triggered strategy above.

\begin{theorem}\label{th4}
Under Theorem \ref{th3} holds, the MASs are without Zeno behavior occurring.
\end{theorem}
\begin{proof}
Following the proof of Theorem \ref{th2} and the triggering function (\ref{equation:42.1}), for each leader to finally yield
\begin{align}\label{equation:54}
t_{k+1}^{i}-t_{k}^{i}
\geq&\ \frac{1}{\bar{A}_{3i}}\ln(\frac{\mu_{3}\bar{A}_{3i}}{\Delta_{1i}}e^{-\varpi_{3}t_{k+1}^{i}}+1)>0,
\end{align}
where $\bar{A}_{3i}=\max_{t\in[t_{k}^{i},t_{k+1}^{i})}\big\{\frac{\dot{\gamma}_{i}}{\gamma_{i}}\!+2\|A\|\!+3\big\}$ and $\Delta_{3i}=\max_{t\in[t_{k}^{i},t_{k+1}^{i})}\{\gamma_{i}\|(A+\!BK)\hat{z}_{i}\|^{2}+\gamma_{i}\|FC\delta_{i}\|^{2}+\gamma_{i}\|Bv_{i}\|^{2}\}$.

For each follower, it follows equally well (\ref{equation:42.2}) that
\begin{align}\label{equation:55}
t_{k+1}^{i}-t_{k}^{i}
\geq&\ \frac{1}{\bar{A}_{4i}}\ln(\frac{\mu_{4}\bar{A}_{4i}}{\Delta_{4i}}e^{-\varpi_{4}t_{k+1}^{i}}+1)>0,
\end{align}
where $\bar{A}_{4i}=\max_{t\in[t_{k}^{i},t_{k+1}^{i})}\big\{\frac{\dot{\varphi}_{i}}{\varphi_{i}}+2\|A\|+3,\dot{\phi}_{i}+\dot{d}_{i}+1\big\}$ and $\Delta_{4i}=\max_{t\in[t_{k}^{i},t_{k+1}^{i})}\{\phi_{i}d^{3}_{i}\|\hat{\theta}_i\|^{2}+\varphi_{i}\|(A+BK)\hat{z}_{i}\|^{2}+\varphi_{i}\|FC\delta_{i}\|^{2}+\varphi_{i}\|B\hat{y}_{i}\|^{2}\}$.
Thus, no Zeno behavior exists for each agent. The proof is done.
\end{proof}
\begin{figure}[!] 
		\centering \includegraphics[width=0.45\textwidth,height=5cm]{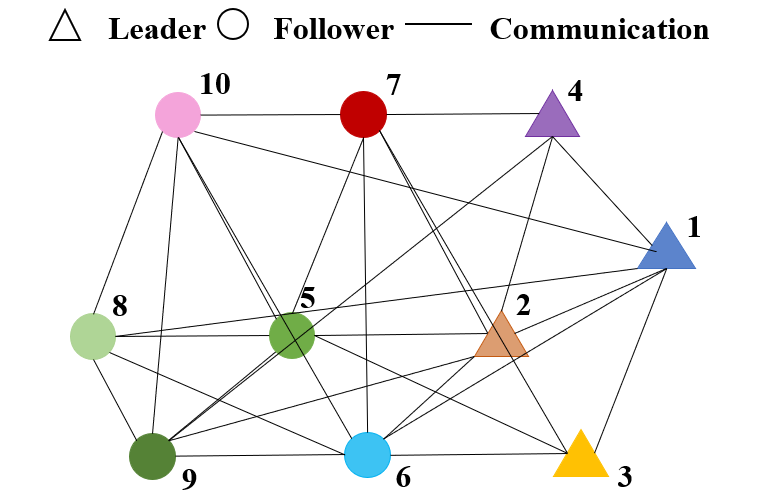}
\caption{The communication graph of agents.}
\label{FIG:1}
\end{figure}

\section{Numerical Simulations \label{Simulations}}

This section verifies the previous theoretical results via numerical simulation experiments.

\begin{table}[!]
\centering
\caption{\\Transformation parameters}
\setlength{\tabcolsep}{1mm}{
\renewcommand\arraystretch{1.2}
\setlength{\arraycolsep}{0.3mm}
\begin{threeparttable}
  \begin{tabular}{cp{2cm}cc}
\toprule
Figure&Transformation& $\Gamma(t)$&$b(t)$\\
\midrule
\ref{fig:2.51}, \ref{fig:3.51}&Rotation\centering&$\begin{bmatrix}
0&0&-1\\
     0&1&0\\
     1&0&0\end{bmatrix}$,
     $\begin{bmatrix}
c&0&-s&0\\
0&0&0&0\\
     s&0&c&0\\
    0&0&0&0\end{bmatrix}$&$\begin{bmatrix}
0\\ 0\\ 0\end{bmatrix}$, $\begin{bmatrix}
0\\ 0\\0\\ 0\end{bmatrix}$ \\
\ref{fig:2.52}, \ref{fig:3.52}&Scale up\centering&$\begin{bmatrix}
2&0&0\\
     0&2&0\\
     0&0&2\end{bmatrix}$, $\begin{bmatrix}
2&0&0&0\\
0&0&0&0\\
     0&0&2&0\\
    0&0&0&0\end{bmatrix}$&$\begin{bmatrix}
0\\ 0\\ 0\end{bmatrix}$, $\begin{bmatrix}
0\\ 0\\0\\ 0\end{bmatrix}$\\
\ref{fig:2.53}, \ref{fig:3.53}&Shear\centering&$\begin{bmatrix}
1&1&0\\
     0&1&0\\
     0&0&1\end{bmatrix}$, $\begin{bmatrix}
1&0&0&0\\
0&0&0&0\\
     1&0&1&0\\
    0&0&0&0\end{bmatrix}$&$\begin{bmatrix}
0\\ 0\\ 0\end{bmatrix}$, $\begin{bmatrix}
0\\ 0\\0\\ 0\end{bmatrix}$\\
\ref{fig:2.54}, \ref{fig:3.54}&\makecell*[c]{Coplanar\\/Colinear}&$\begin{bmatrix}
1&0&0\\
     0&1&0\\
     0&0&0\end{bmatrix}$, $\begin{bmatrix}
3&0&1&0\\
0&0&0&0\\
     3&0&1&0\\
    0&0&0&0\end{bmatrix}$&$\begin{bmatrix}
0\\ 0\\ 0\end{bmatrix}$, $\begin{bmatrix}
0\\ 0\\0\\ 0\end{bmatrix}$\\
\ref{fig:2.55},\ref{fig:3.55}&\makecell*[c]{Rotation, \\scale down,\\ shear, and\\ translation}&$\begin{bmatrix}
0&\frac{1}{2}&-\frac{1}{2}\\
     0&\frac{1}{2}&0\\
     \frac{1}{2}&0&0\end{bmatrix}$, $\begin{bmatrix}
\frac{c}{2}&0&-\frac{s}{2}&0\\
0&0&0&0\\
     \frac{c+s}{2}&0&\frac{c-s}{2}&0\\
    0&0&0&0\end{bmatrix}$&$\begin{bmatrix}
2\\ -2\\ 2\end{bmatrix}$, $\begin{bmatrix}
-2\\ 0\\2\\ 0\end{bmatrix}$\\
\bottomrule
\end{tabular}
 \begin{tablenotes}
     \footnotesize
        \item Note: $c=\cos\frac{\pi}{4}$, $s=\sin\frac{\pi}{4}$.
      \end{tablenotes}
\end{threeparttable}}
\label{tab:1}
\end{table}

\textbf{Case 1}: For the state-based ETC case, we consider the nominal MAS configuration involving four leaders and six followers in the three-dimensional space as $P(a)=[2, 4, 0; -2, -4, 0; -2, 2, 4; 2, -2, 4; -2, 4, 0; -4, 0, 0; 4, 0, 0; \\2, -4, 0; 2, 2, 4; -2, -2, 4]$, whose communication between agents is shown in the Fig. \ref{FIG:1}(a). The linear dynamic of the MAS is depicted by (\ref{equation:1}) with
\begin{equation*}
A=\begin{bmatrix}
		1 & 1 & 0\\
		-1 & 0 & 0\\
        1 & 0 & 1
	\end{bmatrix}, \
B=\begin{bmatrix}
		0 & 1 & 0\\
		-1 & 0 & 0\\
        -1 & 1 & 1
	\end{bmatrix}, 
\end{equation*}
where the eigenvalues of unstable $A$ are $1$ and $0.5\pm0.866\iota$. By setting $R_{1}=I_{d}$, then solving the algebraic Riccati equation in Theorem \ref{th1} yields
\begin{align*}
\setlength{\arraycolsep}{1.2pt}
P=\begin{bmatrix}
		2.20&	0.57&	-0.35\\
0.56&	1.98&	-0.94\\
-0.35&-0.94&1.52
	\end{bmatrix}, {\rm and} \ 
 K=\begin{bmatrix}
0.21&	1.04&	0.58\\
-1.85&	0.39&	-1.18\\
0.35&	0.94&	-1.52
	\end{bmatrix}.
\end{align*}
Then, the compensation term $v_{i}=-B^{-1}(A+BK)p_i^{*}$.

\begin{figure*}[!]
\begin{center}
{
\subfloat[Nominal formation.\label{FIG:2.1}]{\includegraphics[width=5.8cm,height=3cm]{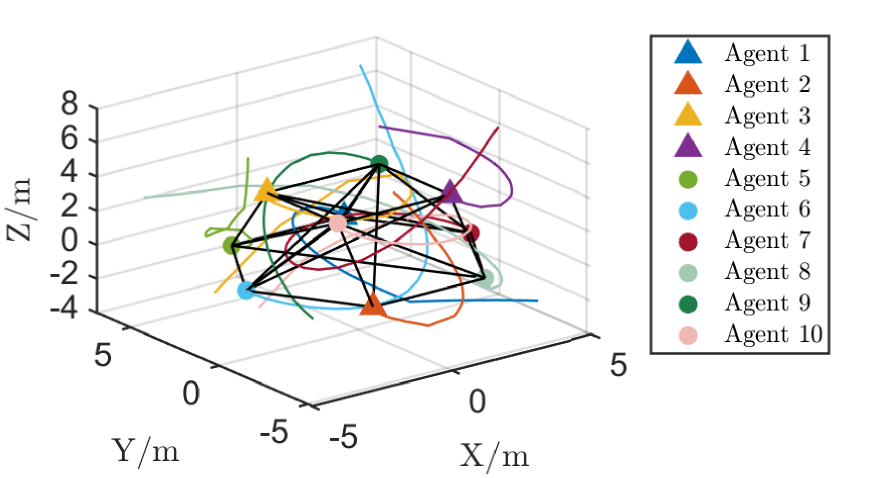}
}
}
{
\subfloat[Rotate 270 degrees on the Y-axis.\label{fig:2.51}]{\includegraphics[width=5.8cm,height=3cm]{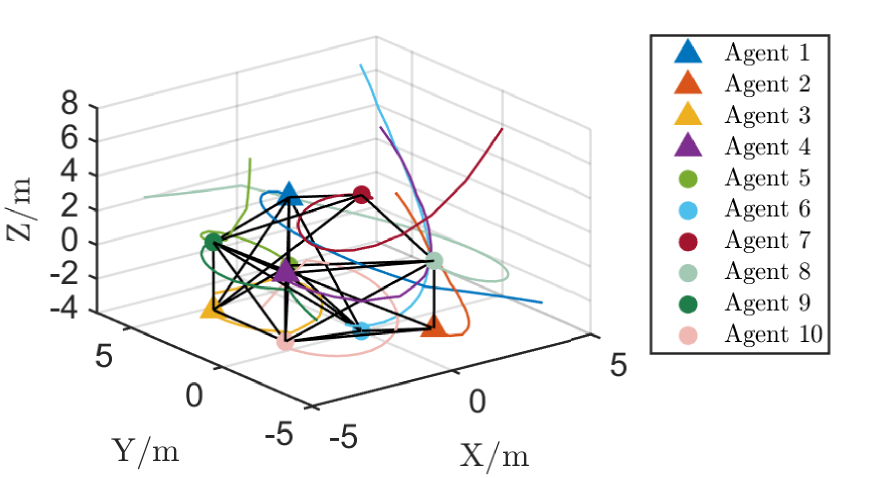}
}
}
{
\subfloat[Enlargement by a twofold factor.\label{fig:2.52}]{\includegraphics[width=5.8cm,height=3cm]{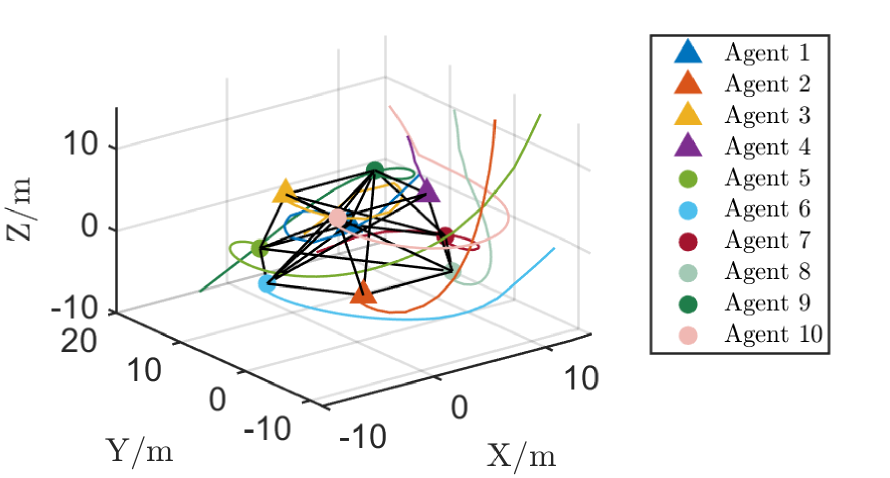}}
}
{
\subfloat[Shear along the X-axis.\label{fig:2.53}]{\includegraphics[width=5.8cm,height=3cm]{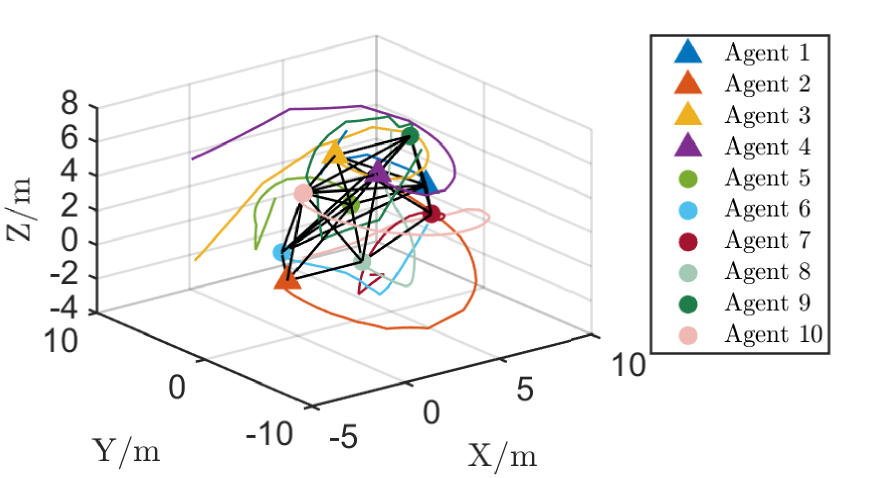}}
}
{
\subfloat[Coplanar. \label{fig:2.54}]{\includegraphics[width=5.8cm,height=3cm]{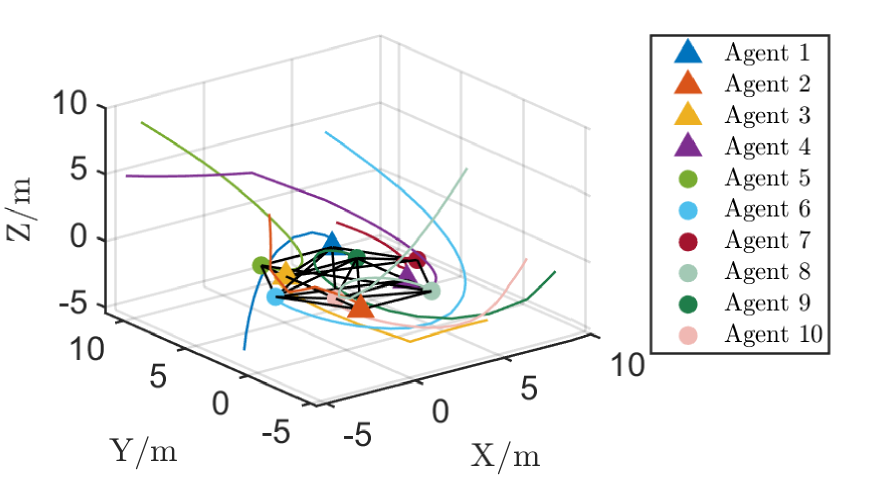}}
}
{
\subfloat[A combination transformation of rotation, scaling down, shear, and translation. \label{fig:2.55}]{\includegraphics[width=5.8cm,height=3cm]{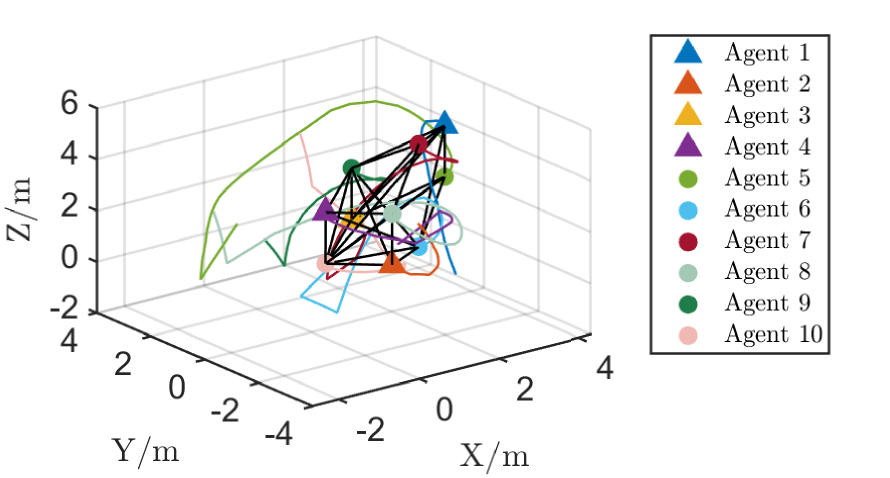}}
}
\caption{The evolutions of agents with nominal formation and some affine transformations under the control scheme (\ref{equation:8}).}
\label{FIG:2.5}
\end{center}
\hrulefill
\end{figure*}

\begin{figure}[!] 
		\centering \includegraphics[width=9cm,height=3cm]{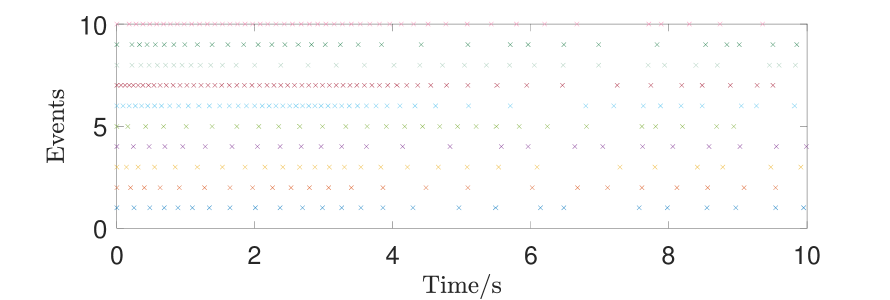}
\caption{The triggering instants under the control scheme (\ref{equation:8}).}
\label{FIG:2.4}
\end{figure}
\begin{figure}[!] 
		\centering \includegraphics[width=9cm,height=4.5cm]{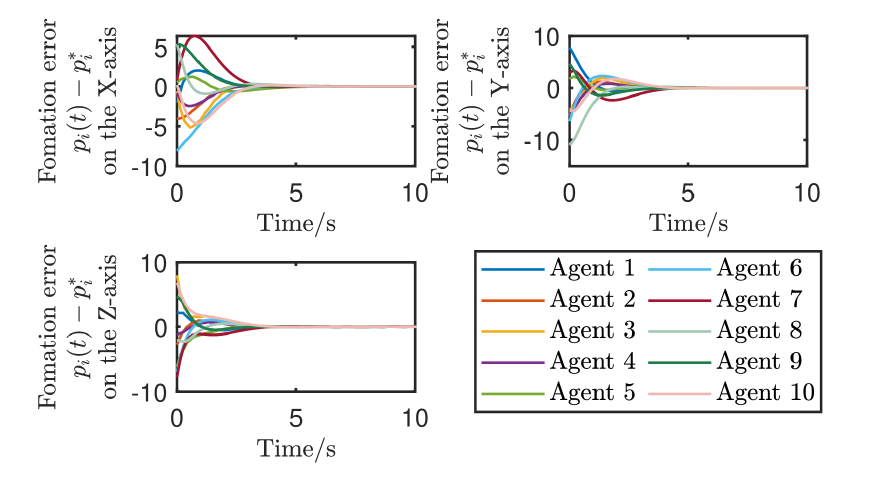}
\caption{The formation errors under the control scheme (\ref{equation:8}).}
\label{FIG:2.2}
\end{figure}

\begin{figure}[!] 
		\centering \includegraphics[width=9.1cm,height=5.5cm]{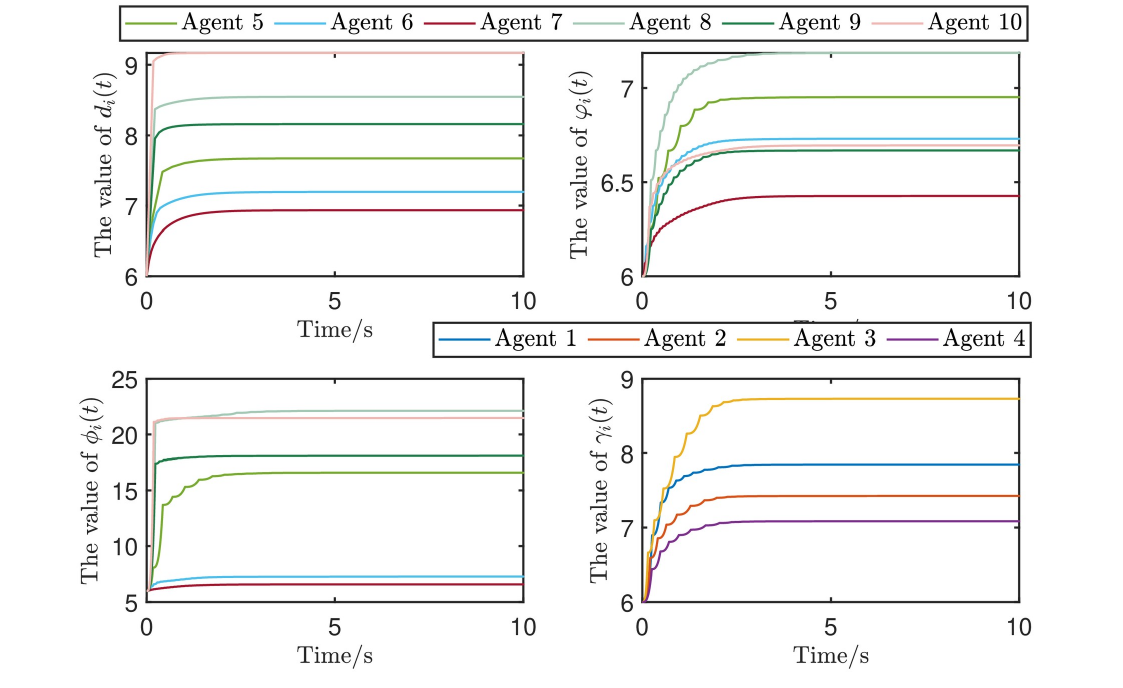}
\caption{The adaptive parameters under the control scheme (\ref{equation:8}).}
\label{FIG:2.3}
\end{figure}


In this example, the given formation and the solved stress matrix $\Omega$ (\ref{equation:70}) by using dynamic programming methods \cite{Zhao2018Affine} can satisfy Assumptions \ref{a1}-\ref{a2}. Set the event-triggered parameters as $\mu_{i}=\varpi_{i}=1$ for $i=1,2$, and the initial values of parameters as $d_{i}(0)=\phi_{i}(0)=\varphi_{i}(0)=6$ for $i\in V_{f}$ and $\gamma_{i}(0)=6$ for $i\in V_{l}$. The initial states are chosen randomly.
\begin{align}\label{equation:70}
\footnotesize{\setlength{\arraycolsep}{0.3pt}
\Omega\!=\!\!\begin{bmatrix}
0.41&	0	&-0.05	&0&	-0.27	&0.05&	-0.10	&0	&-0.09&	0.15\\
0&	0.56&	0&	0&	0	&-0.23&	0&	-0.34	&0.23	&-0.23\\
-0.05&	0&	0.44&	0	&-0.15&	0	&0.20&	0&	-0.24&	-0.20\\
0&	0&	0	&0.36	&0	&0.14& 0&	-0.14&	-0.14	&-0.21\\
-0.27&	0&	-0.15&	0&	0.51	&-0.34&	0	&0.10&	0&	0.15\\
0.05&	-0.23	&0	&0.14	&-0.34&	0.52&	0	&0	&0.06	&-0.21\\
-0.20&	0&	0.20&	0	&0&	0	&0.39&	-0.20	&-0.20&	0\\
0	 & -0.34	&0	&-0.14&	0.10	&0	&-0.20&	0.44	&0	&0.14\\
-0.09&	0.23	&-0.24&	-0.14&	0	&0.06	&-0.20	&0	&0.56&	-0.18\\
0.15	&-0.23	&-0.20&	-0.21&	0.15&	-0.21&	0	&0.14	&-0.18&	0.59
\end{bmatrix}}
\end{align}

To verify the control scheme (\ref{equation:8}), we performed the following experimental simulation on the formation of nominal and their affine-transformed configurations, respectively.

As shown in Fig. \ref{FIG:2.1}, leaders move to the target positions, and the other followers can follow to form the target nominal formation $P(a)$. Then, Fig. \ref{FIG:2.4} shows that the triggering instants become sparser with time. Calculable agents 1-10 have triggering frequencies of $2.6$\%, $2.4$\%, $2.7$\%, $2.5$\%, $2.8$\%, $5$\%, $5.8$\%, $3.4$\%, $3.3$\%, and $4.9$\% for step size $0.01$s, signifying a low communication rate as suitable for real-world scenarios with limited communication. Apart from that, Fig. \ref{FIG:2.2} reveals that the formation errors converge to 0 within $t=10$s, and Fig. \ref{FIG:2.3} displays that all adaptive parameters are stabilized to finite values, confirming the findings given in Theorem \ref{th1}.

\begin{table}[!]
\centering
\caption{\\Events count per agent under the control scheme (\ref{equation:8})\\(Time $10$s, step size $0.01$s)}
\setlength{\tabcolsep}{2.2mm}{
\renewcommand\arraystretch{0.8}
\tabcolsep=0.18cm
\begin{tabular}{p{2cm}cccccccccc}
\toprule
Agent\centering &1&2&3&4&5&6&7&8&9&10\\
\midrule
Nominal\centering&26&24&27&25&28&50&58&34&33&49\\
Rotation\centering&23&24&27&26&38&41&33&51&32&33\\
Scale up\centering&30&32&33&24&74&75&39&44&59&84\\
Shear\centering&23&32&30&30&37&47&41&44&28&64\\
Coplanar\centering&25&25&26&30&34&50&30&43&42&49\\
\makecell*[c]{Rotation, scale\\ down, shear, and \\ translation}&23&21&16&22&45&39&44&41&35&36\\
\bottomrule
\end{tabular}
}
\label{tab:2}
\end{table}

Setting $\Gamma(t)$ and $b(t)$ then performing a series of affine transformations on the nominal configuration, with rotation, scaling, shear, coplanar, and combinatorial transformations, having the movement trajectories shown in Fig. \ref{FIG:2.5}. One can see that the affine formations form steadily with the set transformation parameters (see TABLE \ref{tab:1}). Even more importantly, it is noted in TABLE \ref{tab:2} that almost all agents are statistically less than 9\% of the communication.

\begin{figure*}[!]
\begin{center}
{
\subfloat[Nominal formation.\label{FIG:3.1}]{\includegraphics[width=5.8cm,height=3cm]{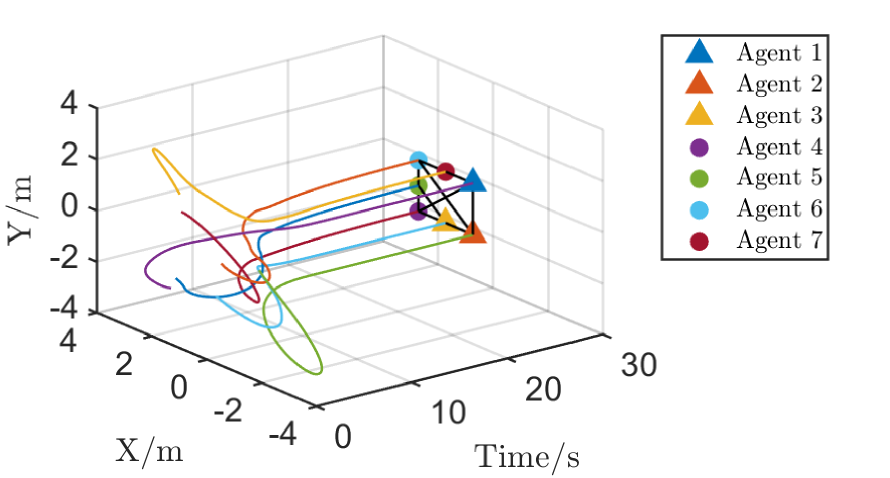}
}
}
{
\subfloat[Rotate 45 degrees clockwise.\label{fig:3.51}]{\includegraphics[width=5.8cm,height=3cm]{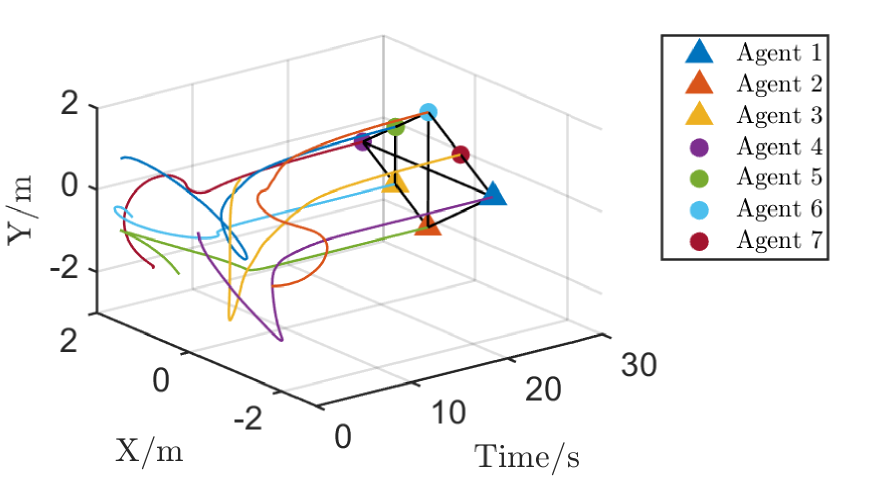}
}
}
{
\subfloat[Enlargement by a twofold factor.\label{fig:3.52}]{\includegraphics[width=5.8cm,height=3cm]{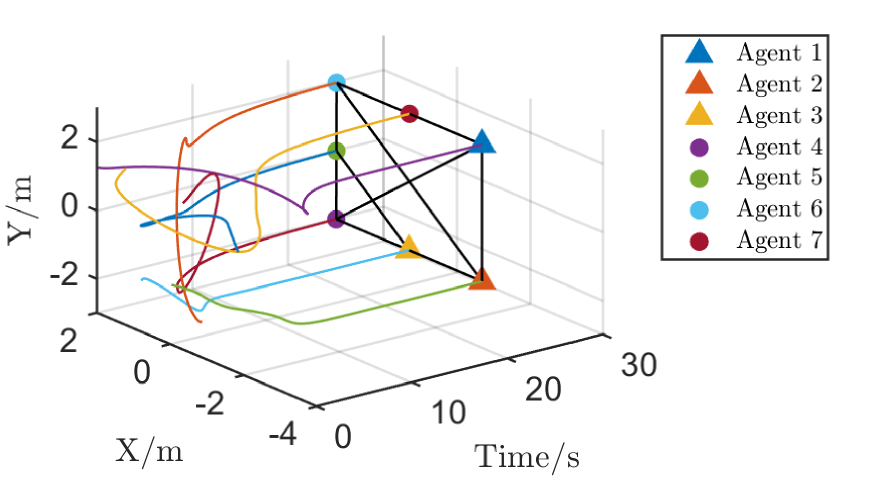}}
}
{
\subfloat[Shear along the Y-axis.\label{fig:3.53}]{\includegraphics[width=5.8cm,height=3cm]{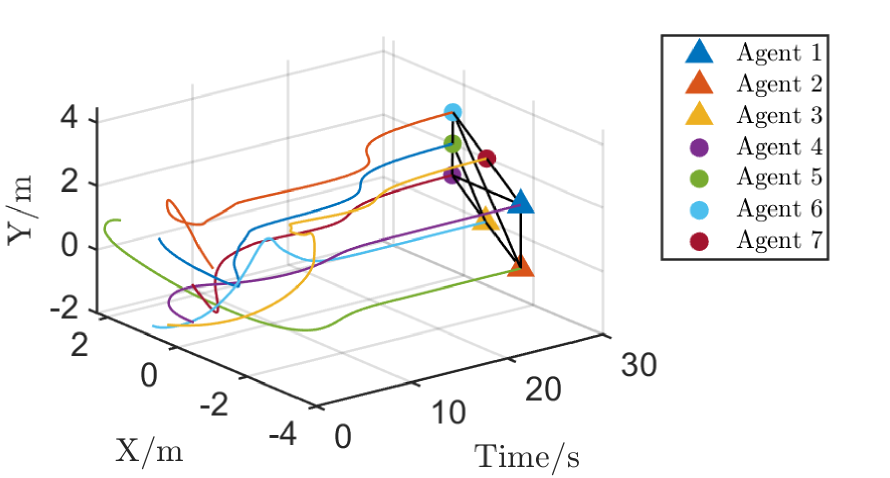}}
}
{
\subfloat[Colinear. \label{fig:3.54}]{\includegraphics[width=5.8cm,height=3cm]{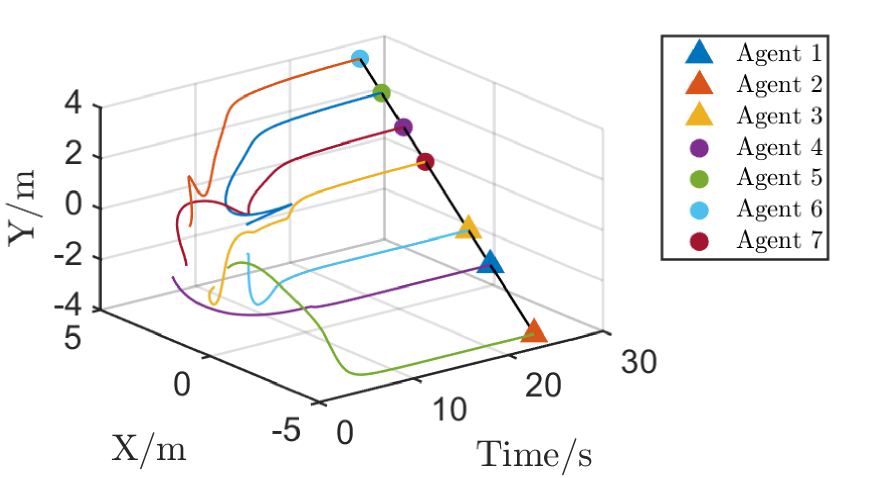}}
}
{
\subfloat[A combination transformation of rotation, scaling down, shear, and translation.
 \label{fig:3.55}]{\includegraphics[width=5.8cm,height=3cm]{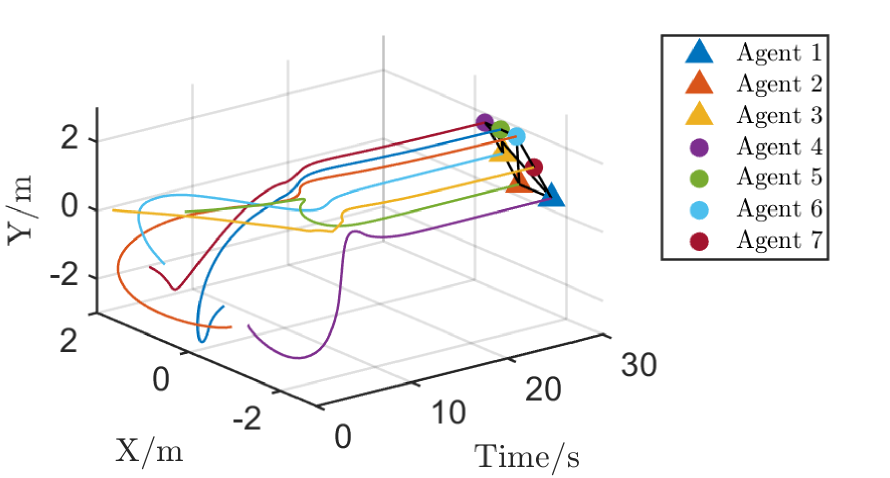}}
}
\caption{The evolutions of agents with nominal formation and some affine transformations under the control scheme (\ref{equation:44}).}
\label{FIG:3.6}
\end{center}
\hrulefill
\end{figure*}

\begin{figure}[!] 
		\centering \includegraphics[width=9cm,height=3cm]{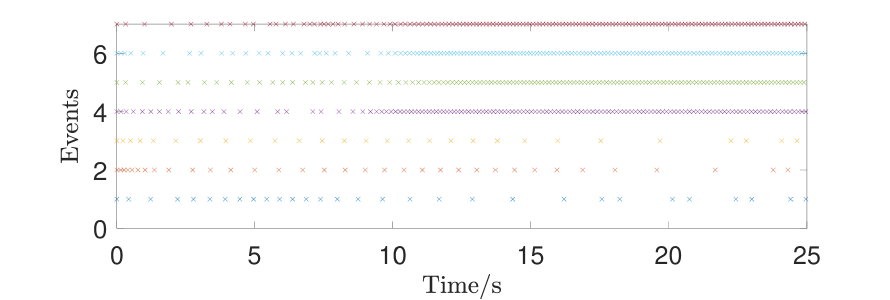}
\caption{The triggering instants under the control scheme (\ref{equation:44}).}
\label{FIG:3.5}
\end{figure}

\begin{figure}[!] 
		\centering \includegraphics[width=9cm,height=6cm]{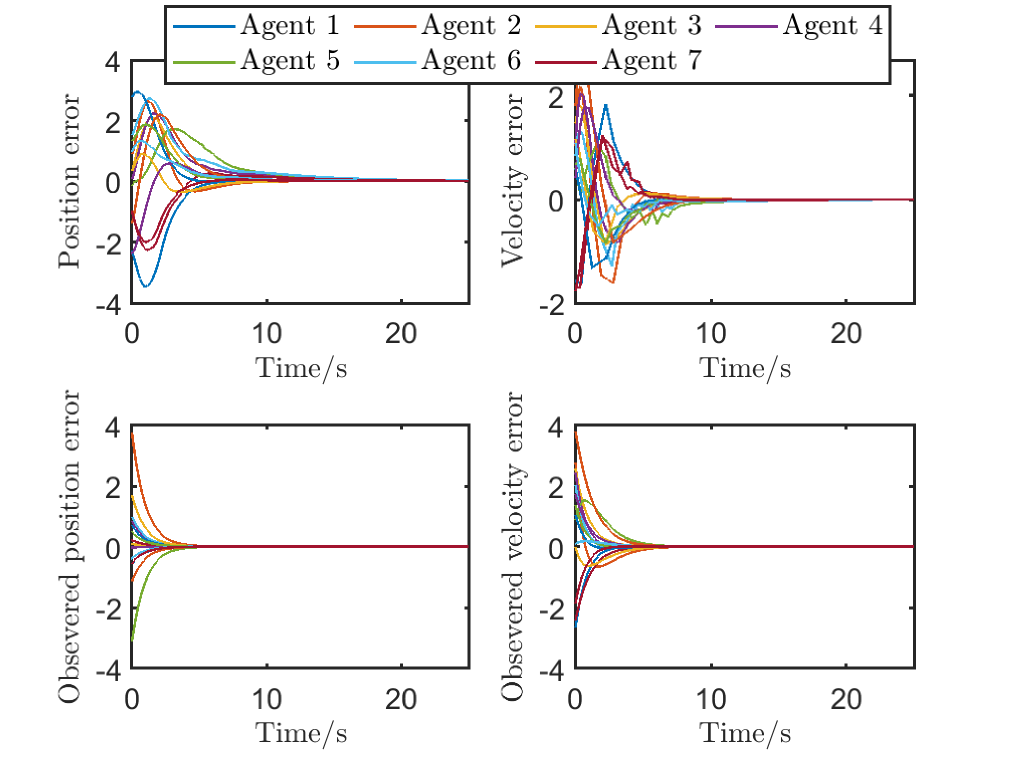}
\caption{The formation errors and observer errors under the \\control scheme (\ref{equation:44}).}
\label{FIG:3.3}
\end{figure}

\begin{figure}[!] 
		\centering \includegraphics[width=9.1cm,height=6cm]{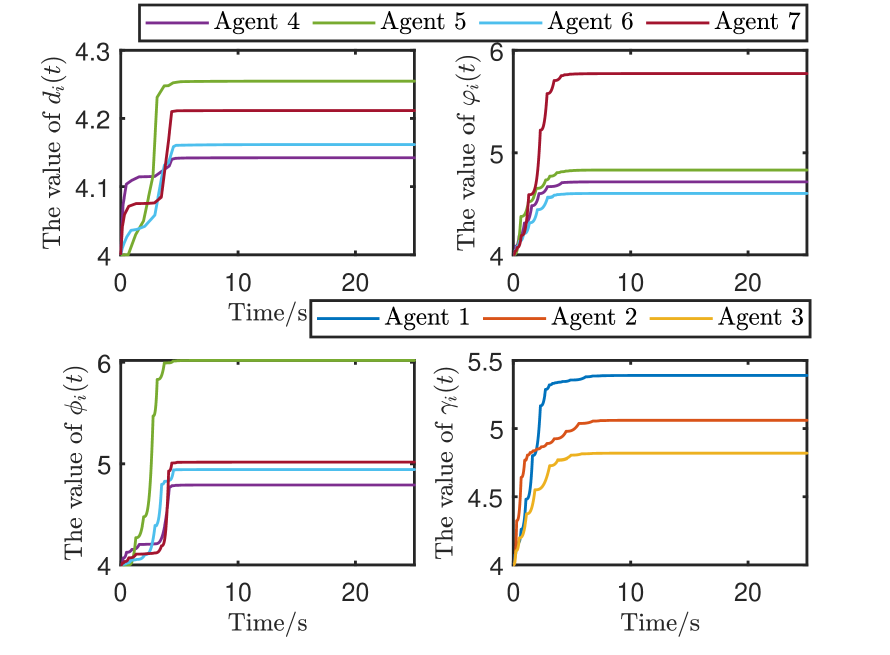}
\caption{The adaptive parameters under the control scheme (\ref{equation:44}).}
\label{FIG:3.4}
\end{figure}

\textbf{Case 2}: In this case, to validate the output-based control scheme (\ref{equation:44}), we employ an MAS consisting of three leader and four follower unicycles:
\begin{equation}\label{equation:56}
\begin{cases}
\dot{x}_{i}(t)=\bar{v}_{i}(t)\cos(\theta_{i}(t)),   \\
\dot{y}_{i}(t)=\bar{v}_{i}(t)\sin(\theta_{i}(t)), \\
\dot{\theta}_{i}(t)=\bar{\upsilon}_{i} (t), \quad i\in\{1,2,\dots,7\},
\end{cases}
\end{equation}
where $x_{i}(t)\in\mathbb{R}$, $y_{i}(t)\in\mathbb{R}$ are the Cartesian coordinates of the center of MASs, $\theta_{i}(t)\in\mathbb{R}$ is the heading angle in the inertial frame, $\bar{v}_{i}(t)\in\mathbb{R}$ and $\bar{\upsilon}_{i}(t)\in\mathbb{R}$ are the linear and angular velocities of $i$-th vehicle.

Following the dynamic feedback linearization method of Reference \cite{10164160} and assuming the output model, the model (\ref{equation:56}) can be transformed into the system (\ref{equation:1}) with parameters
\begin{equation*}
A=\begin{bmatrix}
		0 & 1 & 0 & 0\\
		0 & 0 & 0  & 0\\
            0 & 0 & 0 & 1\\
            0 & 0 & 0  & 0
	\end{bmatrix}, \
B=\begin{bmatrix}
		0 & 0 \\
		1 & 0 \\
            0 & 0 \\
            0 & 1
            \end{bmatrix}, 
C^{T}=\begin{bmatrix}
		1 & 0 \\
		1 & 0 \\
            0 & 1 \\
            0 & 1
            \end{bmatrix}
\end{equation*} 
and the state $p_{i}(t)=[x_{i}(t),\bar{v}_{x_{i}}(t),y_{i}(t),\bar{v}_{y_{i}}(t)]^{T}$, the input $u_{i}(t)=[u_{x_{i}}(t),u_{y_{i}}(t)]^{T}$, where $\bar{v}_{x_{i}}(t)=\bar{v}_{i}(t)\cos(\theta_{i}(t))$, $\bar{v}_{y_{i}}(t)=\bar{v}_{i}(t)\sin(\theta_{i}(t))$.
Following Theorem \ref{th3} and setting $R_{1}=R_{2}=I_{d}$, the gain matrices $F$ and $K$ can be obtained.

In this case, define the nominal configuration as $P(a)=[-1,0,1,0; -1,0,-1,0; 0,0,-1,0; 1,0,-1,0; 1,0,0,0; 1,0,\\1,0; 0,0,1,0]$. Fig. \ref{FIG:1}(b) on page 7 is the topology graph and the obtained positive eigenvalues are $\lambda(\Omega_{ff})=\{0.77, 0.88, 1.15, 1.43\}$. By Remark \ref{r3}, the matrices $U_{1}=[0, 1, 0, 0; 0, 0, 0, 1]$ and $U_{2}=[1, 0, 0, 0; 0, 0, 1, 0]$ make $U_{2}Ap^{*}_{i}=0$ hold and $v_{i}=-U_{1}(A+BK)p^{*}_{i}$.
Set $\varpi_{i}=0.6$, $\mu_{i}=1$ for $i=3,4$,  $d_{i}(0)=\phi_{i}(0)=\varphi_{i}(0)=4$ for $i\in V_{f}$ and $\gamma_{i}(0)=4$ for $i\in V_{l}$.

\begin{table}[!]
\centering
\caption{\\Event count per agent under the control scheme (\ref{equation:44})\\(Time $25$s, step size $0.01$s)}
\setlength{\tabcolsep}{2.2mm}{
\renewcommand\arraystretch{0.8}
\tabcolsep=0.22cm
\begin{tabular}{p{2.5cm}ccccccc}
\toprule
Agent\centering&1&2&3&4&5&6&7\\
\midrule
Nominal\centering&30&34&28&136&131&141&150\\
Rotation\centering&27&30&28&139&148&147&143\\
Scale up\centering&32&24&28&135&124&132&131\\
Shear\centering&27&30&26&124&120&132&135\\
Colinear\centering&31&31&29&131&161&140&130\\
\makecell*[c]{Rotation, scale down, \\shear, and  translation}&35&26&31&134&134&149&140\\
\bottomrule
\end{tabular}
}
\label{tab:3}
\end{table}

Parallel to this, the affine transformations of the nominal formation with the parameters in TABLE \ref{tab:1} on page 7 are also validated. Figs. \ref{FIG:3.6}-\ref{FIG:3.3} on page 8 verify stable affine transformations under finite communication and partial state availability. From TABLE \ref{tab:3}, each agent has a communication triggering rate of no more than 6\% in per transformation.

\section{Conclusions \label{Conclusions}}

In this paper, we studied the distributed event-triggered AFC for linear MASs, which relies only on the state information at the triggering instants, thus significantly reducing the communication frequency and the controller update rate. Tracking controllers are designed for leaders, which removes the usual assumption on system matrix. The designed adaptive ETC protocols not requiring the global information, can be extended to MASs of disparate scales and communication links. Furthermore, the output-based control protocol is also suitable for state unavailability scenarios, expanding its applications.

In unknown contingency situations, such as the battlefield environment, formation maneuvering operations are very critical for breakout and obstacle avoidance tasks. The setting of the time-invariant desired position will result in various static affine formations. In future work, the extension of the method to apply to dynamic formation maneuvering control is necessary. Meanwhile, to handle more practical and generalized situations, we hope to investigate directed network topologies.


\bibliographystyle{IEEEtran}
\bibliography{reference}
\end{document}